\documentclass[journal]{IEEEtran}

\pdfoutput=1

\usepackage{graphicx}
\usepackage{amsmath}
\usepackage{amsfonts}
\usepackage{amsthm}
\usepackage{amssymb}
\usepackage{cite}				
\usepackage{xcolor}
\usepackage{microtype}
\usepackage{authblk}
\usepackage{mathtools}
\usepackage{algorithm}
\usepackage{algpseudocode}
\usepackage{cleveref}
\usepackage{subcaption}
\usepackage{bm}

\newtheorem{thm}{Theorem}
\newtheorem*{thm*}{Theorem}

\newtheorem{defn}{Definition}

\newtheorem*{exm*}{Example}
\newtheorem{asm}{Assumption}

\newcommand{\rr}{\mathrm}
\newcommand{\cc}{\mathcal}
\newcommand{\bb}{\mathbb}	
\newcommand{\ff}{\mathfrak}
\newcommand{\bo}{\bm}
\newcommand{\bs}{\boldsymbol}

\IEEEoverridecommandlockouts

\makeatletter
\newcommand\fs@spaceruled{\def\@fs@cfont{\bfseries}\let\@fs@capt\floatc@ruled
	\def\@fs@pre{\vspace{0.5\baselineskip}\hrule height.8pt depth0pt \kern2pt}%
	\def\@fs@post{\kern2pt\hrule\relax}%
	\def\@fs@mid{\kern2pt\hrule\kern2pt}%
	\let\@fs@iftopcapt\iftrue}
\makeatother

\begin{document}

\title{On Differential Privacy for Federated Learning in Wireless Systems with Multiple Base Stations}

\author{
	Nima Tavangaran,~\IEEEmembership{Member,~IEEE,}
    Mingzhe Chen,~\IEEEmembership{Member,~IEEE,}
	Zhaohui Yang,~\IEEEmembership{Member,~IEEE,}\\
	Jos\'e Mairton B. Da Silva Jr.,~\IEEEmembership{Member,~IEEE,}
	and H. Vincent Poor,~\IEEEmembership{Fellow,~IEEE}%
	\thanks{N. Tavangaran, J. M. B. da Silva Jr., and H. V. Poor are with the Department of Electrical and Computer Engineering, Princeton University, Princeton, NJ, USA, (e-mails: nimat@princeton.edu and poor@princeton.edu).}%
	\thanks{M. Chen is with the Department of Electrical and Computer Engineering and Institute for Data Science and Computing, University of Miami, Coral Gables, FL, 33146 USA (e-mail: mingzhe.chen@miami.edu).}
	\thanks{Z. Yang is with the College of Information Science and Electronic Engineering, Zhejiang University, Hangzhou 310027, China, (e-mail: yang\_zhaohui@zju.edu.cn).}
	\thanks{J. M. B. da Silva Jr. is also with the School of Electrical Engineering and Computer Science, KTH Royal Institute of Technology, Stockholm, Sweden (e-mail: jmbdsj@kth.se).}
	\thanks{The work of N. Tavangaran was partly supported by the German Research Foundation (DFG) under Grant TA 1431/1-1.}%
	\thanks{The work of H. V. Poor was supported by the U.S. National Science Foundation under Grants CCF-1908308 and CNS-2128448.}%
	\thanks{This work has been submitted to the IEEE for possible publication.  Copyright may be transferred without notice, after which this version may no longer be accessible.}
}

\maketitle

\begin{abstract}
In this work, we consider a federated learning model in a wireless system with multiple base stations and inter-cell interference. We apply a differential private scheme to transmit information from users to their corresponding base station during the learning phase. We show the convergence behavior of the learning process by deriving an upper bound on its optimality gap. Furthermore, we define an optimization problem to reduce this upper bound and the total privacy leakage. To find the locally optimal solutions of this problem, we first propose an algorithm that schedules the resource blocks and users. We then extend this scheme to reduce the total privacy leakage by optimizing the differential privacy artificial noise. We apply the solutions of these two procedures as parameters of a federated learning system. In this setting, we assume that each user is equipped with a classifier. Moreover, the communication cells are assumed to have mostly fewer resource blocks than numbers of users. The simulation results show that our proposed scheduler improves the average accuracy of the predictions compared with a random scheduler. Furthermore, its extended version with noise optimizer significantly reduces the amount of privacy leakage.
\end{abstract}


\section{Introduction}

Machine Learning (ML) systems are expected to play an important role in future mobile communication standards~\cite{eldar2022machine}. With increasing applications of ML schemes in wireless systems, new technologies are emerging to enhance the performance of such systems. On the other hand the wireless technology itself can also be deployed to enhance the ML procedures~\cite{hellstrom2022wireless}. Among possible candidates, Federated Learning (FL) has been shown to have considerable promise~\cite{bonawitz2019towards,konevcny2016federated,mcmahan2017communication} and has the potential to benefit from wireless communication.

FL solves several issues of centralized ML systems by distributing the learning task among several edge devices. One advantage of using an FL system, which makes it a good fit in a wireless setting, is that edge devices do not need to transmit their local datasets to the server. This reduces the amount of wireless resources that is required for accomplishing the given ML task. Apart from this, the privacy of each edge device is not completely compromised since the server does not have a direct access to the data~\cite{li2020federated}.

FL schemes operating over wireless networks have
been extensively researched in recent years; see for example~\cite{samarakoon2019distributed,chen2020joint,amiri2020federated,yang2019eeFL,hamdi2021federated,wang2022int,chen2022fedw}. In~\cite{chen2020joint}, the authors studied the effects of wireless parameters on the FL process. They
derived an upper bound on the optimality gap of the convergence terms and proposed an optimization problem to minimize the upper bound by considering wireless parameters like resource allocation, user scheduling, and packet error rate. Some other works that studied the communication aspects of FL are~\cite{zheng2020energy,wang2019adaptive,tran2019federated}. Moreover, FL with several layers of aggregation or with hierarchy has been studied in~\cite{hosseinalipour2020multistage,khan2021disp, khan2021disp2,zhang2022scalable,pandey2022edge,asad2022fl}.

Although the training data of each device in FL is not transmitted to the server, yet a function of the local model (query) is still sent to the server. It has been shown that this local model might leak some information about the training data~\cite{al2016reconstruction}. To mitigate this drawback, FL has been extensively studied together with a privacy preserving scheme called Differential Privacy (DP)~\cite{dwork2006calibrating}. 

DP-based schemes follow the principle of not being adversely affected much by having one's data used in any analysis~\cite{dwork2014algorithmic}. This powerful notion is well established and is applied in industry. To realize a DP-based FL system, each edge device adds some artificial noise to its transmitting information. This noise provides a certain amount of privacy depending on the noise power and sensitivity of the query function. 

DP based FL and its convergence behavior have been extensively studied; see for example~\cite{hu2020concentrated,wu2021incentivizing,sun2021pain,wei2020federated,wei2021low,seif2021privacy,liu2022fl}. In this regard, the work in~\cite{hu2020concentrated} addresses the privacy implementation challenges through a combination of zero-concentrated
differential privacy, local gradient perturbation and secure aggregation.\cite{wei2021low} considers the resource allocation and user scheduling to minimize the FL training delay under the constraint of performance and DP requirements. Finally,~\cite{liu2022fl} presents a closed-form global loss and privacy leakage of a DP-based FL system and then minimizes the loss and privacy leakage.

However, none of these works consider a joint learning and resource allocation scheme for DP-based FL that considers the effects of inter-cell interference as well. 
In this paper, we adapt the framework in~\cite{chen2020joint} and consider a wireless FL system in a multiple base station scenario. Additionally, we consider DP noise added to the gradients~\cite{hu2020concentrated} and combine this approach with resource scheduling.

The goal of the FL system here is to train a global model for a given predictor. We introduce an iterative DP-based FL scheme with two levels of aggregation (Algorithm~\ref{alg_dfl}) and then derive an upper bound on the optimality gap of its convergence terms (Theorem~\ref{thm_con}).

We then propose an optimization problem whose goal is to improve the convergence of the upper bound on the optimality gap of Algorithm~\ref{alg_dfl} and simultaneously reduce the total privacy leakage. In this regard, the optimization problem is with respect to certain variables like user and resource scheduling, uplink transmit powers, and the amount of DP noise that is applied by each user to its transmitting information.

Since the proposed optimization problem is a non-linear multi-variable mixed integer programming, we divide it into two simpler schemes. 

First, we present a suboptimal approach to minimize the objective function only with respect to resource scheduling variables in a sequential manner, i.e., cell by cell. This reduces the original problem to a linear integer programming task and substantially simplifies the implementation. We call this scheme also the optimal scheduler (OptSched). Since this approach performs sequentially from one cell to another one, the amounts of optimal transmit power should be adjusted carefully due to the effects of inter-cell interference. To tackle this problem, we introduce a procedure to determine the users' optimal transmit powers by solving a simple optimization problem. 

Next, we enhance the OptSched scheme by further minimizing the objective function of the proposed optimization problem with respect to the DP noise. This leads us to a convex optimization problem with respect to the DP noise standard deviations. We call this extended scheme the optimal scheduler with DP optimizer (OptSched+DP).

We present all the numerical optimizations and benchmarking results. In this regard, we apply Python optimization packages like CVXPY, CVXOPT, GLPK, and ECOS~\cite{diamond2016cvxpy,agrawal2018rewriting,cvxopt,glpk,Domahidi2013ecos}. The numerical results show that our proposed schemes (OptSched and OptSched+DP) reduce the objective function of the optimization task substantially compared with the case in which we randomly allocate the resources and apply the DP noise. 

Next, we apply these (sub-)optimal parameters to our iterative learning scheme (Algorithm~\ref{alg_dfl}). In this regard, each user is equipped with a fully connected multi-layer neural network as a classifier. Furthermore, we assume that communication cells have mainly more users than available resource blocks. This is a legitimate assumption due to the bandwidth limitation. We then perform simulations to measure the accuracy, loss, and the amount of the privacy leakage in such a system for the proposed algorithms. To realize the simulations, we apply the TensorFlow, NumPy, and Matplotlib   packages~\cite{tensorflow2015-whitepaper,harris2020array,Hunter:2007}.

The simulations show that the OptSched scheme predominantly improves the classification accuracy by scheduling the users who have larger data chunks and better uplink channels. The OptSched+DP scheme, on the other hand, achieves a significant reduction in privacy leakage of individual users by systematically adjusting the DP noise power and moderately sacrificing accuracy.

\vspace{3px}
\emph{Notation:}
We denote vectors by lowercase bold letters, e.g. $\bo{w}$. Matrices are represented by uppercase bold letters like $\bo{X}$, or the identity matrix $\bo{I}_d$ with $d$ rows and $d$ columns. Sets are denoted by Calligraphic fonts like $\cc{X}$. Random mechanisms as a special kind of functions are represented by Fraktur fonts, e.g. $\ff{M}$. The transpose of a vector $\bo{x}$ is denoted by $\bo{x}^\intercal$. Logarithms are assumed to be to the basis 2. The set of real numbers is represented by $\bb{R}$. $[R]$ denotes the set $\{1,2,\ldots,R\}$.

\section{System Model}\label{sec_mod}

We begin this section by reviewing some preliminary notions on DP that are required in this work. The complete list of definitions can be found in~\cite{dwork2006calibrating,dwork2014algorithmic,bun2016concentrated}.

\subsection{Differential Privacy Model}
Let a data universe $\cc{X}$ and the distribution $P_X$ on it be given. Assume that a database is denoted by a matrix $\bo{X}\in\cc{X}^{K\times m}$ and contains $K$ rows of independent and identically distributed (i.i.d.) $m$-dimensional samples (row vectors). Two databases $\bo{X},\tilde{\bo{X}}\in\cc{X}^{K\times m}$ are called adjacent if they differ only in one row. 

A query (mechanism) $q:\cc{X}^{K\times m}\rightarrow \bb{R}^d$ is a function which takes a database $\bo{X}\in\cc{X}^{K\times m}$ as input and gives a $d$-dimensional output. If the output of the query contains randomness then it is called a randomized mechanism.

In the following, we introduce the notion of privacy for randomized mechanisms, which are defined on a given set of databases $\cc{X}^{K\times m}$.

\begin{defn}
	A randomized mechanism $\ff{M}:\cc{X}^{K\times m}\rightarrow\bb{R}^d$ is said to be $(\epsilon,\delta)$-Differentially Private, or for short $(\epsilon,\delta)$-DP, if for every adjacent $\bo{X},\tilde{\bo{X}}\in\cc{X}^{K\times m}$, we have that
	\begin{align}
		\rr{Pr}\big(\ff{M}(\bo{X})\in\cc{W}\big)\leq \rr{e}^{\epsilon}\,\rr{Pr}\big(\ff{M}(\tilde{\bo{X}})\in\cc{W}\big)+\delta
	\end{align}
	holds for any $\cc{W}\subset\bb{R}^d$.
\end{defn}

In this work, we apply a relaxed version of the $(\epsilon,\delta)$-DP that is more suitable for Gaussian mechanisms.

\begin{defn}
	A randomized mechanism $\ff{M}:\cc{X}^{K\times m}\rightarrow\bb{R}^d$ is said to be $\rho$-zero-Concentrated Differentially Private (CDP), or for short $\rho$-zCDP, if
	\begin{align}
		\rr{D}_{\alpha}\big(\ff{M}(\bo{X})\,\|\,\ff{M}(\tilde{\bo{X}})\big)\leq \rho\alpha
	\end{align}
	holds for every adjacent $\bo{X},\tilde{\bo{X}}\in\cc{X}^{K\times m}$  and all $\alpha\in(1,\infty)$, where $\rr{D}_{\alpha}$ is the $\alpha$-R\'enyi divergence~\cite{bun2016concentrated}.
\end{defn}

\subsection{Federated Learning Model}

Based on the notions from previous section, we introduce our privacy preserving FL model for a system with multiple base stations. Let a collection of  base stations denoted by the set $\cc{S}$ be given such that they can communicate with each other through a main server. Assume that each base station $s\in\cc{S}$ serves a set of edge devices (users) denoted by $\cc{U}_s$, where the users in $\cc{U}_s$ have some arbitrary order. Let $U_s$ denote the size of this set.

We assume that each user $i\in\cc{U}_s$ assigned to the base station $s$ has access to a database
\begin{align*}
	\bo{X}_{s,i}&\coloneqq\big(\bo{x}_{s,i}^{(1)},\bo{x}_{s,i}^{(2)},\ldots,\bo{x}_{s,i}^{(K_{s,i})}\big)^\intercal\in\bb{R}^{K_{s,i}\times m},
\end{align*}
where $K_{s,i}$ is the number of samples (row vectors) in the database $\bo{X}_{s,i}$. Each row of the above matrix, say $\bo{x}_{s,i}^{(k)}$, is an $m$-dimensional data sample given by
\begin{align*}
	\bo{x}_{s,i}^{(k)}\coloneqq\big(x_{s,i}^{(k)}(1),x_{s,i}^{(k)}(2),\ldots,x_{s,i}^{(k)}(m-1),y_{s,i}^{(k)}\big),
\end{align*}
where the first $m-1$ elements are the inputs and the last entry $y_{s,i}^{(k)}$ is the output of the training data.

In the first step of the FL scheme at round $t=1$, the main server broadcasts a weight vector $\bo{w}^{(t)}\in\bb{R}^d$ to all base stations.
This vector is called the global model and can be initialized randomly. Then, each base station $s$ transmits this model to all of its edge devices. Let only a subset of the users in each cell are active and participate in the learning process. 
We denote the active users in cell $s$ by:
\begin{align}
	\bo{a}_s\coloneqq (a_{s,i})_{i\in\cc{U}_{s}}\in\{0,1\}^{U_{s}},\label{eqn_sch}
\end{align}
where $a_{s,i}=1$ indicates that user $i\in\cc{U}_s$ is scheduled~\cite{chen2020joint} to participate in the learning process and $a_{s,i}=0$, otherwise.

\begin{figure}[!t]
	\centering
	\huge
	\scalebox{0.5}{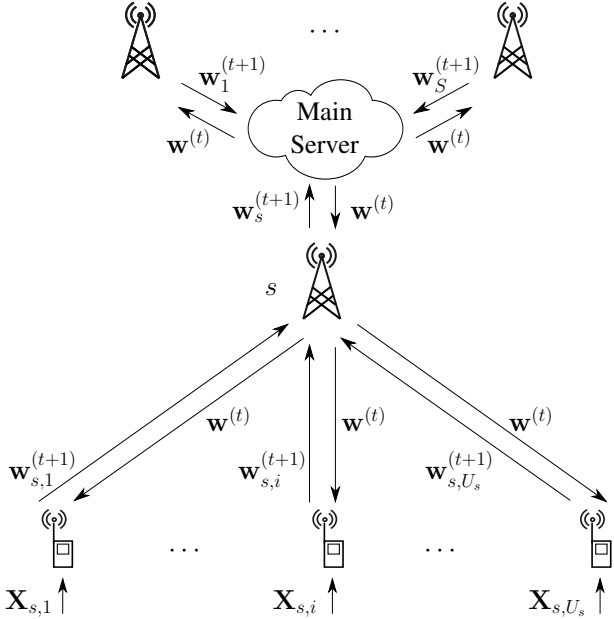}
	\caption{FL model with multiple base stations.}
	\label{fig_model}
\end{figure}
Fig.~\ref{fig_model} shows an example of a model with multiple base stations. In this example, all users of cell $s$ (depicted on the bottom of the figure), are scheduled to participate in the FL process and receive the vector $\bo{w}^{(t)}$.

Each scheduled user computes a local loss function depending on the ML algorithm that is applied in the system. We denote the loss function of a user $i\in\cc{U}_s$ by $l(\bo{w}^{(t)}, \bo{x}_{s,i}^{(k)})$, which is a function of the global model and its training sample. Next, this user computes the gradient~\cite{nesterov2003introductory,konevcny2016federated} of its loss function over all given samples as a query function
\begin{align}
	q_{s,i}^{(t)}(\bo{X}_{s,i})\coloneqq \frac{1}{K_{s,i}}\sum_{k=1}^{K_{s,i}}\nabla l(\bo{w}^{(t)}, \bo{x}_{s,i}^{(k)}),\label{eqn_sdf}
\end{align}
where the gradients are with respect to $\bo{w}^{(t)}$. 

Similarly as in~\cite{hu2020concentrated}, the user then applies Gaussian noise $\bo{n}^{(t)}_{s,i}\sim\cc{N}(\bo{0},\sigma_{s,i}^2\bo{I}_d)$ to the outcome of the query to implement the randomized mechanism $\ff{M}_s^{(t)}$ as follows
\begin{align}
	\ff{M}_{s,i}^{(t)}(\bo{X}_{s,i})\coloneqq q_{s,i}^{(t)}(\bo{X}_{s,i})+\bo{n}^{(t)}_{s,i}.\label{eqn_upd}
\end{align}
In this context, $\bo{n}^{(t)}_{s,i}$ is assumed to be independent of all other random variables in our model, including the DP noise that is applied in previous iterations. The main reason of applying Gaussian noise is that it gives tight bounds when applied with zCDP~\cite{bun2016concentrated}.
The following vector denotes the noise standard deviations of the users in the cell $s$:
\begin{align}\label{eqn_qlm}
	\bs{\sigma}_s\coloneqq (\sigma_{s,i})_{i\in\cc{U}_s}.
\end{align}

Next, the edge device $i\in\cc{U}_s$ updates its local model by
\begin{align}
	\bo{w}_{s,i}^{(t+1)}\coloneqq\bo{w}^{(t)}-\lambda\ff{M}_{s,i}^{(t)}(\bo{X}_{s,i}),\label{eqn_lcm}
\end{align}
where $\lambda>0$ is the learning step size. It then transmits its updated model\footnote{For the sake of simplicity, we consider only a single local update at each edge device in each round and also consider the batch gradient descent.} $\bo{w}_{s,i}^{(t+1)}$ to its corresponding base station $s$.

In the next step, base station $s$ aggregates all received updated models $\bo{w}_{s,i}^{(t+1)}$ as given below:
\begin{align}
	\bo{w}_s^{(t+1)}\coloneqq \frac{1}{\sum_{i\in\cc{U}_s}K_{s,i}a_{s,i}}\sum_{i\in\cc{U}_s}K_{s,i}a_{s,i}\bo{w}_{s,i}^{(t+1)},\label{eqn_agg}
\end{align}
where $a_{s,i}$ is the scheduling parameter and was defined as an element of the vector $\bo{a}_s$ in Eq.~\eqref{eqn_sch}.

Consequently, all base stations send their aggregated models to the main server. There, the global model at round $t+1$ is computed as follows
\begin{align}
\bo{w}^{(t+1)}\coloneqq \frac{1}{K_{\rr{a}}}\sum_{s\in\cc{S}}\sum_{i\in\cc{U}_s}K_{s,i}a_{s,i}\bo{w}_s^{(t+1)},\label{eqn_ag2}
\end{align}
where 
\begin{align}\label{eqn_tr0}
K_{\rr{a}}\coloneqq\sum_{s\in\cc{S}}\sum_{i\in\cc{U}_s}K_{s,i}a_{s,i}
\end{align} 
is the total number of training samples of all scheduled users.

\floatstyle{spaceruled}
\restylefloat{algorithm}
\begin{algorithm}[t]
	\caption{Privacy preserving FL with multiple stations}\label{alg_dfl}
	\begin{algorithmic}[1]
		\State The main server broadcasts $(\bo{a}_s,\bs{\sigma}_s)_{s\in\cc{S}}$, which are given by~\eqref{eqn_sch} and~\eqref{eqn_qlm}, to all base stations and their users.\\ 
		The main server initializes the global model $\bo{w}^{(0)}$.
		\For{$t=0:T$}
		\State The main server broadcasts $\bo{w}^{(t)}$ to all base stations.
		\For{ base stations $s\in\cc{S}$ in parallel}
		\State Base station $s$ broadcasts $\bo{w}^{(t)}$ to all its users.
		\For{users $i\in\cc{U}_s$ in parallel}
		\If{$a_{s,i}=1$}
		\State The user $i\in\cc{U}_s$ updates its model as in~\eqref{eqn_lcm}.
		\State The user $i\in\cc{U}_s$ then sends $\bo{w}_{s,i}^{(t+1)}$ back to the base station $s$.
		\EndIf
		\EndFor
		\State The base station $s$ aggregates the received models as in~\eqref{eqn_agg}.
		\State The base station $s$ then sends $\bo{w}_s^{(t+1)}$ back to the main server.
		\EndFor
		\State The main server aggregates all models as in~\eqref{eqn_ag2}.
		\EndFor
	\end{algorithmic}
\end{algorithm}

Next, the main server broadcasts the new global model $\bo{w}^{(t+1)}$ to the base stations where it is then forwarded further to their corresponding users. This process continues for a given number of $T$ iterations. Algorithm~\ref{alg_dfl} summarizes these steps, where $(\bo{a}_s,\bs{\sigma}_s)_{s\in\cc{S}}$ are assumed to be shared with all participants at the beginning of the learning process.

One difference between Algorithm~\ref{alg_dfl} and other approaches like e.g. the FL schemes in~\cite{hu2020concentrated,chen2020joint} is that here the aggregation is done in two steps. Additionally, the noise standard deviations $\sigma_{s,i}$ at users are not necessarily identical here and a joint optimal user scheduling and DP noise adjustment is possible.

To characterize the DP noise, we need to make the following assumption which can be achieved in practice by weight clipping~\cite{abadi2016dldp,hu2020concentrated}.

\begin{asm}\label{asm_inp}
	The gradients of the local loss functions are always upper bounded:
	\begin{align*}
		\|\nabla l(\bo{w}, \bo{x})\|_2 \leq L.
	\end{align*}
\end{asm}

In~\cite{hu2020concentrated}, it was shown that if Assumption~\ref{asm_inp} holds, then after $T$ iterations, a mechanism like $\ff{M}_{s,i}^{(t)}(\bo{X}_{s,i})$ is $\rho$-zCDP where
\begin{align}
	\rho = 2T\left(\frac{L}{K_{s,i}\sigma_{s,i}}\right)^2\label{eqn_sli}
\end{align}
is the privacy leakage.

Note that the DP noise affects the convergence of the learning process as well. In Section~\ref{sec_con}, we study the convergence behavior of Algorithm~\ref{alg_dfl} with respect to the vector parameters $\bo{a}_s$ and $\bs{\sigma}_s$. 

\section{Convergence Analysis}
\label{sec_con}
In the following, we define the global loss as a function of local losses. We then derive an upper bound on the optimality gap that appears in each round of Algorithm~\ref{alg_dfl}.

The global loss function is computed over all base stations and is given by
\begin{align}
f(\bo{w}^{(t)})\coloneqq\frac{1}{K}\sum_{s\in\cc{S}}\sum_{i\in\cc{U}_s}
\sum_{k=1}^{K_{s,i}}l(\bo{w}^{(t)}, \bo{x}_{s,i}^{(k)}),\label{eqn_bms}
\end{align}
where $K=\sum_{s\in\cc{S}}\sum_{i\in\cc{U}_s}K_{s,i}$ is the total number of samples (including scheduled or non-scheduled).

The following assumptions are necessary to analyze the global loss function $f$ and have been used before in the literature~\cite{friedlander2012hybrid}.

\begin{asm}\label{eqn_khs}
	The loss function $f:\bb{R}^d\rightarrow \bb{R}$ has a minimum value, i.e., there exists an input vector $\bo{w}^*=\arg\min_{\bo{w}\in\bb{R}^d}(f(\bo{w}))$.
\end{asm}

\begin{asm}\label{asm_ulp}
	The gradient $\nabla f (\bo{w})$ is uniformly $L$-Lipschitz continuous with respect to the model $\bo{w}$, i.e.,
	\begin{align*}
	\|\nabla f(\bo{w})-\nabla f(\bo{w}')\|_2\leq L\|\bo{w}-\bo{w}'\|_2\quad\text{ for all }\bo{w},\bo{w}'\in\bb{R}^d.
	\end{align*}
\end{asm}

\begin{asm}\label{asm_scx}
	The loss function $f:\bb{R}^d\rightarrow \bb{R}$ is $\mu$-strongly convex, i.e.,
	\begin{align*}
	f(\bo{w})\geq f(\bo{w}')+(\bo{w}-\bo{w}')^\intercal\nabla f(\bo{w}')+\frac{1}{2}\mu\|\bo{w}-\bo{w}'\|_2^2
	\end{align*}
	holds for all $\bo{w},\bo{w}'\in\bb{R}^d$.
\end{asm}

\begin{asm}\label{asm_twc}
	The loss function $f:\bb{R}^d\rightarrow \bb{R}$ is twice continuously differentiable. Then, Assumptions~\ref{asm_ulp} and~\ref{asm_scx} are equivalent to the following:
	\begin{align*}
	\mu\bo{I}_d\preceq \nabla^2 f(\bo{w})\preceq L\bo{I}_d\qquad\text{ for all }\bo{w}\in\bb{R}^d.
	\end{align*}
\end{asm}

\begin{asm}\label{asm_tw1}	
	There exists constants $\xi_1\geq 0$ and $\xi_2\geq 1$, such that for any training sample $\bo{x}$ and model $\bo{w}\in\bb{R}^d$, the following inequality holds
	\begin{align*}
	\|\nabla l(\bo{w}, \bo{x})\|_2^2\leq\xi_1+\xi_2\|\nabla f(\bo{w})\|_2^2.
	\end{align*}
\end{asm}

Now, we are ready to derive an upper bound on the optimality gap of Algorithm~\ref{alg_dfl}.

\begin{thm}\label{thm_con}
	Let Assumptions~\ref{eqn_khs}--\ref{asm_tw1} hold. Then, 
	the following upper bound on the optimality gap for Algorithm~\ref{alg_dfl} holds:
	\begin{align*}
		\bb{E}\big[f(\bo{w}^{t+1})-f(\bo{w}^*)\big]\leq C_1\bb{E}\Big[f(\bo{w}^{(t)})-f(\bo{w}^*)\Big]+C_2
		+C_3,
	\end{align*}
	where the expectation is taken over the DP noise and
	\begin{align*}
	C_1&=1-\frac{\mu}{L}+\frac{4\xi_2}{K^2}\bigg(
	\sum_{s\in\cc{S}}\sum_{i\in\cc{U}_s}K_{s,i}
	\big(1-a_{s,i}\big)\bigg)^2,\nonumber\\
	C_2&=\frac{2\xi_1}{LK^2}
	\bigg(
	\sum_{s\in\cc{S}}\sum_{i\in\cc{U}_s}K_{s,i}
	\big(1-a_{s,i}\big)\bigg)^2,\nonumber\\
	C_3&=\frac{d}{2L}\sum_{s\in\cc{S}}
	\sum_{i\in\cc{U}_s}\Big(\frac{K_{s,i}a_{s,i}}{K_{\rr{a}}}\sigma_{s,i}\Big)^2.
	\end{align*}
\end{thm}
\begin{proof}
	The proof is provided in the Appendix.
\end{proof}

Theorem~\ref{thm_con} shows that the expected difference between the global loss and the optimal value $f(\bo{w}^*)$ per iteration is upper bound by expressions that depend on $C_1$, $C_2$, and $C_3$. Hence, by lowering the values of $C_1$, $C_2$, and $C_3$, the convergence of Algorithm~\ref{alg_dfl} should be improved.
In addition, Theorem~\ref{thm_con} shows that the upper bound on the optimality gap is influenced by the scheduling parameters $a_{s,i}$ and DP noise standard deviations $\sigma_{s,i}$. We also observe that the upper bound converges only if $C_1<1$. In Section~\ref{sec_wir}, we design an optimal scheduler and a DP optimizer based on these variables and their effect on this upper bound.

\section{Learning over Wireless Channels with Inter-cell Interference}\label{sec_wir}
In this section, we consider other wireless parameters of the communication system and connect them to the notion of learning. These wireless parameters include resource allocation, transmit power consumption, fading channels, inter-cell interference, and communication rate.

We assume that the users apply an Orthogonal Frequency-Division Multiple Access (OFDMA) technique in the uplink channel to transmit data to their corresponding base station. In this case, each edge device $i\in\cc{U}_s$ is assigned a resource block indexed by $n\in[R]$ where $R$ is the total number of available uplink transmission resource blocks in each cell.

We define the uplink resource allocation matrix $\bo{R}_s$ in a given cell $s$ as
\begin{align}\label{eqn_mat}
\bo{R}_s\coloneqq
\big(r_{s,i}^{(1)}, r_{s,i}^{(2)}, \dots, r_{s,i}^{(R)}\big)_{i\in\cc{U}_{s}},
\end{align}
where $r_{s,i}^{(n)}\in\{0,1\}$. Each row of this matrix represents the resource allocation for a user $i\in\cc{U}_s$. In this case, $r_{s,i}^{(n)}=1$ indicates that the edge device $i\in\cc{U}_s$ uses resource block $n$ in the uplink transmission and $r_{s,i}^{(n)}=0$, otherwise. Moreover, we assume that each active user ($a_{s,i}=1$) is assigned only one resource block and inactive users ($a_{s,i}=0$) are not assigned any resource block at all, i.e.,
\begin{align}\label{opt_22}
\sum_{n=1}^{R} r_{s,i}^{(n)}=a_{s,i}, \forall i\in\cc{U}_s,s\in\cc{S}.
\end{align}

In addition, edge devices in a given cell $s$ do not interfere with each other, i.e.,
\begin{align}\label{eqn_sfi}
	\sum_{i\in\cc{U}_s}r_{s,i}^{(n)} \leq 1,\forall n\in[R],s\in\cc{S}.
\end{align}

To be able to formulate the communication rate, we need first to define the transmit powers of the users. Let the uplink transmit power vector of all edge devices at a given cell $s\in\cc{S}$ be denoted by
\begin{align*}
\bo{p}_s\coloneqq (p_{s,i})_{i\in\cc{U}_s},
\end{align*}
where $p_{s,i}$ denotes the transmit power of the user $i\in\cc{U}_s$. Moreover, the maximum transmit power of each user in any cell is denoted by $P_{\max}$. 

\begin{figure}[!t]
	\centering
	\Large
	\scalebox{0.62}{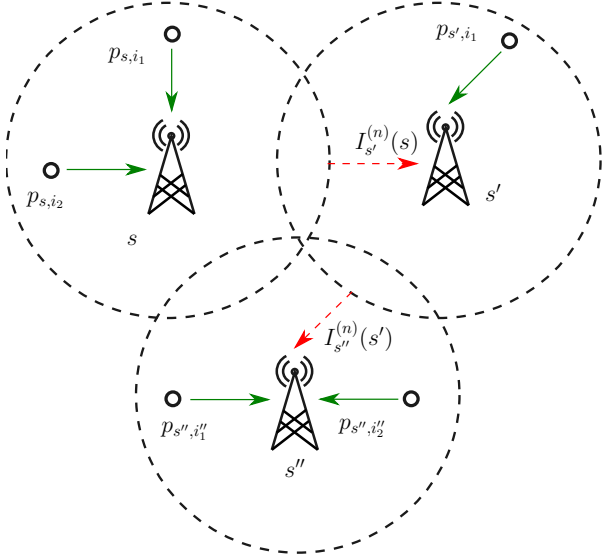}
	\caption{Uplink stage of the wireless FL model with multiple base stations and inter-cell interference.}
	\label{fig_base}
\end{figure}

Another wireless parameter, which is of great importance in the considered system with multiple base stations, is the inter-cell interference.
Let $I_{s}^{(n)}(\tilde{s})$ denote the interference signal power~\cite{moretti2007ofdm} from the cell $\tilde{s}\in\cc{S}\backslash\{s\}$ that affects the uplink signal received by the base station $s$ on the resource block $n$. In this case, inequality~\eqref{eqn_sfi} implies that $I_{s}^{(n)}(\tilde{s})$ is a factor of the transmit power of only one user in the cell $\tilde{s}$ that transmits signals on the resource block $n$. In other words, the received interference signal power can be formulated as
\begin{align}\label{eqn_int}
    I_{s}^{(n)}(\tilde{s}) = \sum_{i\in\cc{U}_{\tilde{s}}} h_{s,i} r_{\tilde{s},i}^{(n)}p_{\tilde{s},i}.
\end{align}
The term $h_{s,i}$ in~\eqref{eqn_int} is the channel gain between the user $i\in\cc{U}_{\tilde{s}}$ and the base station $s$ and can be computed by determining the pathloss~\cite{rappaport1995pro}. The channel gain is given by
\begin{align}
	h_{s,i} = l^2\Big(\frac{c}{4\pi f}\Big)^2\Big(\frac{1}{d_{s,i}}\Big)^3,\label{eqn_cgn}
\end{align}
where $f$ is the uplink center frequency, $d_{s,i}$ is the distance between the user $i\in\cc{U}_{\tilde{s}}$ and base station $s$, $l$ is the output of a Rayleigh distribution with a unit scale parameter, and $c$ is the speed of light.

Fig.~\ref{fig_base} illustrates an example of a wireless communication system with three base stations $s,s',$ and $s''$ in the uplink stage. In this example, the received signals on the resource block $n$ at the base station $s'$ are affected by the interference signal power $I_{s'}^{(n)}(s)$ from cell $s$. Furthermore, base station $s''$ is affected by the interference signal power $I_{s''}^{(n)}(s')$ from cell $s'$.

Let the uplink fading channel between each user $i\in\cc{U}_s$ and its corresponding base station $s$ be fixed and equal to $h_{s,i}$. Also assume that the uplink bandwidth is denoted by $B$. Furthermore, all participants are assumed to have perfect channel knowledge. It is known (see e.g.~\cite{moretti2007ofdm,chen2020joint}) that the maximum uplink communication rate between each user $i\in\cc{U}_s$ and its corresponding base station $s$ can be formulated as
\begin{align}\label{eqn_uls}
c_{s,i}^\rr{U}&\coloneqq \sum_{n=1}^{R}r_{s,i}^{(n)}B\log
\bigg(1+\frac{p_{s,i}h_{s,i}}{\sum_{\tilde{s}\in\cc{S}\backslash\{s\}}I_{s}^{(n)}(\tilde{s})+BN_0}\bigg),
\end{align}
where $N_0$ is the thermal noise power spectral density. We assume that the minimum required uplink communication rate between each user and its base station is denoted by a constant $R_{\min}$.

Before applying these wireless parameters in the learning process, we first need to introduce another measure based on the DP standard deviations $(\bs{\sigma}_s)_{s\in\cc{S}}$. In this regard, we define the total privacy leakage as follows: 
\begin{align}\label{eqn_soi}
2TL^2\sum_{s\in\cc{S}}\sum_{i\in\cc{U}_s}\big(\frac{1}{K_{s,i}\sigma_{s,i}}\big)^2 a_{s,i}.
\end{align}
In this definition, the summands in~\eqref{eqn_soi} are computed by multiplying the privacy leakage given by~\eqref{eqn_sli} at each user with the scheduling variables $a_{s,i}$.
Minimizing this measure reduces the individual privacy leakage as well. This improvement is due to the systematic adjustment of the DP noise power at each user. In this case, edge devices who have a larger number of samples $K_{s,i}$ are assigned less DP noise power.

The parameters $\bo{R}_s$ and $\bs{\sigma}_s$ play a critical role in improving the convergence rate of Algorithm~\ref{alg_dfl} (cf. Theorem~\ref{thm_con}) and reducing the total privacy leakage. Furthermore, the parameter $\bo{p}_s$ is critical in establishing a reliable communication. By minimizing $C_1$ and $C_2$ with respect to these parameters, the upper bound on the optimality gap in Theorem~\ref{thm_con} reduces and thus the convergence rate of the FL procedure should improve. To this end, it is sufficient to minimize only the expressions inside the squared term in $C_1$. 

Therefore, we propose an optimization problem over the variables $(\bo{R}_s,\bo{p}_s,\bs{\sigma}_{s})_{s\in\cc{S}}$ and minimize the values of $C_1$ and $C_2$ from Theorem~\ref{thm_con} and the total privacy leakage given by~\eqref{eqn_soi}. In this combined formulation, we assume that other FL parameters, such as $L,\mu,\xi_1,\xi_2,d,$ and $T$, are constant. The main server can then solve this optimization problem and then broadcast the results to all base stations before Algorithm~\ref{alg_dfl} starts. 

Since it is hard to directly solve a multi-objective optimization problem for both scheduling and total privacy leakage, we formulate the problem as a single-objective optimization task as follows:
\begin{align}
&\underset{(\bo{R}_s,\bo{p}_s,\bs{\sigma}_{s})_{s\in\cc{S}}}{\text{minimize}}
\quad\sum_{s\in\cc{S}}\sum_{i\in\cc{U}_s}K_{s,i}
\bigg[1-\sum_{n=1}^{R} r_{s,i}^{(n)}\bigg]
\nonumber \\
& \qquad\qquad\qquad\qquad +\gamma\sum_{s\in\cc{S}}\sum_{i\in\cc{U}_s}\big(\frac{1}{K_{s,i}\sigma_{s,i}}\big)^2\sum_{n=1}^{R} r_{s,i}^{(n)}\label{opt_ob1}\\
& \text{subject to}\nonumber\\
& \sum_{s\in\cc{S}}
\sum_{i\in\cc{U}_s}K_{s,i}\sigma_{s,i}^2\sum_{n=1}^{R} r_{s,i}^{(n)}\leq V_{\max}\!\sum_{s\in\cc{S}}\sum_{i\in\cc{U}_s}K_{s,i}\sum_{n=1}^{R} r_{s,i}^{(n)},\label{opt_1}\\
& K_{s,i}\sigma_{s,i}\geq N_{\min}\sum_{n=1}^{R} r_{s,i}^{(n)},&&\hspace{-9em}\forall s\in\cc{S},i\in\cc{U}_s,\label{opt_6}\\
& \sum_{i\in\cc{U}_s}r_{s,i}^{(n)} \leq 1,&&\hspace{-9em}\forall s\in\cc{S}, n\in[R],\label{opt_3}\\
&\sum_{n=1}^{R} r_{s,i}^{(n)} \leq 1\; \text{ and }\;  r_{s,i}^{(n)}\in\{0,1\},&&\hspace{-9em}\forall s\in\cc{S},i\in\cc{U}_s,\label{opt_33}\\
& 0\leq p_{s,i}\leq P_{\max},&&\hspace{-9em}\forall s\in\cc{S},i\in\cc{U}_s,\label{opt_4}\\ & c_{s,i}^\rr{U}\geq R_{\min}\sum_{n=1}^{R} r_{s,i}^{(n)},&&\hspace{-9em}\forall s\in\cc{S},i\in\cc{U}_s,\label{opt_5}
\end{align}
where $\gamma>0$ is a constant and is used to balance the optimization of the scheduling and the total privacy leakage. Typically, the value of the constant $\gamma$ can be obtained by hyperparameter tuning and simulations. 

Minimizing the first term in the objective function in~\eqref{opt_ob1} improves the convergence of Algorithm~\ref{alg_dfl} and is computed by applying~\eqref{opt_22} to the summation term in $C_1$ of Theorem~\ref{thm_con}. Minimizing the second term, on the other hand, reduces the total privacy leakage at all users and is given by~\eqref{eqn_soi}.

Constraint~\eqref{opt_1} guarantees that the DP noise error, which is characterized by the term $C_3$ of Theorem~\ref{thm_con}, is less than a given constant $V_{\max}$. To derive condition~\eqref{opt_1}, we first consider the following upper bound on the squared term in $C_3$ 
\begin{align}
\Big(\frac{K_{s,i}a_{s,i}}{K_{\rr{a}}}\sigma_{s,i}\Big)^2\leq\frac{K_{s,i}a_{s,i}}{K_{\rr{a}}}\sigma_{s,i}^2,
\end{align}
which follows by~\eqref{eqn_tr0} and $K_{s,i}a_{s,i}\leq K_{\rr{a}}$. Constraint~\eqref{opt_1} then follows by applying~\eqref{eqn_tr0} and~\eqref{opt_22} to this upper bound and setting it to be smaller than $V_{\max}$. We can then control the amount of DP noise variance and its error by adjusting the constant $V_{\max}$. 

Conditions~\eqref{opt_3}-\eqref{opt_33} provide the resource allocation constraints, whereas~\eqref{opt_4}-\eqref{opt_5} restrict the transmit power to a maximum amount $P_{\rr{max}}$ and ensure a minimum communication rate $R_{\rr{min}}$ for each user in each cell, respectively.
Finally, constraint~\eqref{opt_6} guarantees an upper bound on the privacy leakage of the users individually due to~\eqref{eqn_sli}. In this case, the constant $N_{\min}$ controls the minimum amount of DP noise at each user. 

We notice that the variable $\bo{a}_s$ and $\bo{R}_s$, which are given by~\eqref{eqn_sch} and~\eqref{eqn_mat}, are related due to~\eqref{opt_22}. Therefore, $\bo{a}_s$ does not appear as a minimization variable. 

The optimization problem in~\eqref{opt_ob1} is 	not easy to solve. However, we can subdivide it into simpler problems and search for (sub-)optimal solutions. The main server can then compute and broadcast these (sub-)optimal $(\bo{R}_s^*,\bo{p}_s^*,\bs{\sigma}_{s}^*)_{s\in\cc{S}}$ to all base stations where they can be forwarded to the users. These computations and initialization should be done prior to the beginning of Algorithm~\ref{alg_dfl}.

\section{Algorithm Design}
\label{sec_alg}
In this section, we propose two suboptimal sequential algorithms to solve the optimization problem in~\eqref{opt_ob1}. First, for fixed DP noise the objective function in~\eqref{opt_ob1} is minimized with respect to users' transmit powers and resource block allocation in a cell-by-cell manner. In the second part, with given transmit power and resource block allocation, the optimization problem in~\eqref{opt_ob1} becomes convex with respect to the DP noise standard deviations.

\subsection{Optimal Scheduler} \label{sub_rba}
Let the DP noise standard deviations be given such that condition~\eqref{opt_6} is always satisfied. In the following, we consider the joint transmit power and resource block allocation problem, which is a simplified version of~\eqref{opt_ob1}.
\begin{align}
& \underset{(\bo{R}_s,\bo{p}_s)_{s\in\cc{S}}}{\text{minimize}}
\quad\sum_{s\in\cc{S}}\sum_{i\in\cc{U}_s}K_{s,i}
\bigg[1-\sum_{n=1}^{R}r_{s,i}^{(n)}\bigg]
\nonumber\\
& \qquad\qquad\quad\quad +\gamma\sum_{s\in\cc{S}}\sum_{i\in\cc{U}_s}\big(\frac{1}{K_{s,i}\sigma_{s,i}}\big)^2\sum_{n=1}^{R} r_{s,i}^{(n)}\label{eqn_jww}\\
& \text{subject to~\cref{opt_1,opt_3,opt_33,opt_4,opt_5}}.\nonumber
\end{align}

The optimization problem in~\eqref{eqn_jww} is non-linear with respect to $(\bo{R}_s)_{s\in\cc{S}}$ due to~\eqref{opt_5} and~\eqref{eqn_uls}. To further simplify it, we first compute the optimal transmit powers while guaranteeing the minimum communication rate constraint. In this case, setting~\eqref{opt_5} to equality, combining it with~\eqref{eqn_uls}, and using the fact that $r_{s,i}^{(n)}\in\{0,1\}$, the optimal transmit powers can be obtained as
\begin{align}\label{opt_power}
&p_{s,i}^*=\sum_{n=1}^R r_{s,i}^{(n)}\Big(2^{\frac{R_{\min}}{B}}-1\Big)\frac{\sum_{{\tilde{s}}\in\cc{S}\backslash\{s\}}I_{s}^{(n)}(\tilde{s})+BN_0}{h_{s,i}}.
\end{align}

We then consider only one cell at each optimization step in an alternating strategy. Based on this approach and applying~\eqref{opt_power} to~\eqref{opt_4}, the optimization task in~\eqref{eqn_jww} reduces to the following linear integer programming problem for a single cell $s$:
\begin{align}
&\underset{\bo{R}_s}{\text{minimize}}
\quad\sum_{i\in\cc{U}_s}K_{s,i}
\bigg[1-\sum_{n=1}^{R}r_{s,i}^{(n)}\bigg]\nonumber\\
& \qquad\qquad\qquad\qquad\quad +\gamma\sum_{i\in\cc{U}_s}\big(\frac{1}{K_{s,i}\sigma_{s,i}}\big)^2\sum_{n=1}^{R} r_{s,i}^{(n)}\label{opt_ob3}\\
& \text{subject to}\nonumber\\
&\sum_{i\in\cc{U}_s}K_{s,i}\sigma_{s,i}^2\sum_{n=1}^{R}r_{s,i}^{(n)}+\sum_{{\tilde{s}}\in\cc{S}\backslash\{s\}}
\sum_{i\in\cc{U}_{\tilde{s}}}K_{\tilde{s},i}\sigma_{\tilde{s},i}^2\sum_{n=1}^{R}r_{\tilde{s},i}^{(n)}\nonumber\\
&\;\;\leq V_{\max}
\sum_{i\in\cc{U}_s}K_{s,i}\sum_{n=1}^{R}r_{s,i}^{(n)}+V_{\max}\!\!\!\sum_{{\tilde{s}}\in\cc{S}\backslash\{s\}}
\sum_{i\in\cc{U}_{\tilde{s}}}K_{{\tilde{s}},i}\sum_{n=1}^{R}r_{{\tilde{s}},i}^{(n)},\\
& \sum_{i\in\cc{U}_s}r_{s,i}^{(n)} \leq 1,
\forall n\in[R],\\
&\sum_{n=1}^{R} r_{s,i}^{(n)} \leq 1\; \text{ and }\;  r_{s,i}^{(n)}\in\{0,1\},
\forall i\in\cc{U}_s,\\
&0\leq\sum_{n=1}^R r_{s,i}^{(n)}\Big(2^{\frac{R_{\min}}{B}}- 1\Big)\frac{\sum_{\tilde{s}\in\cc{S}\backslash\{s\}}I_{s}^{(n)}(\tilde{s})+BN_0}{h_{s,i}}
\leq P_{\max},\nonumber\\&\forall i \in\cc{U}_s.
\end{align}

\begin{algorithm}[t]
	\caption{Random scheduler with random DP noise (RndSched)}\label{alg_non}
	\begin{algorithmic}[1]
		\State Initialize the values of $(\bo{R}_s,\bo{p}_s,\bs{\sigma}_{s})_{s\in\cc{S}}$ randomly such that they satisfy~\eqref{opt_1}-\eqref{opt_4}.
		\State Compute $(\bo{p}^*_s)_{s\in\cc{S}}$ by solving~\eqref{eqn_siw} and unschedule those users whose communication rates do not meet~\eqref{opt_5}.
		\State Output the resulting parameters as a (sub-)optimal solution $(\bo{R}_s,\bo{p}^*_s,\bs{\sigma}_{s})_{s\in\cc{S}}$.
	\end{algorithmic}
\end{algorithm}

To solve~\eqref{opt_ob3}, we assume that $(\bo{R}_{\tilde{s}},\bo{p}_{\tilde{s}})_{\tilde{s}\in\cc{S}\backslash\{s\}}$ are known and satisfy conditions~\eqref{opt_3}-\eqref{opt_4}. We then solve this problem with respect to $\bo{R}_s$ while taking $\bo{R}_{\tilde{s}}$ with $\tilde{s}\in\cc{S}\backslash\{s\}$ as constants. By solving this optimization problem for each cell, we obtain a (sub-)optimal scheduling solution $(\bo{R}^*_s)_{s\in\cc{S}}$ for the whole system.

In the next step, the optimal transmit powers $(\bo{p}_s)_{s\in\cc{S}}$ should be accordingly computed by using~\eqref{opt_power}. Yet, the term $\sum_{\tilde{s}\in\cc{S}\backslash\{s\}}I_{s}^{(n)}(\tilde{s})$ in~\eqref{opt_power} is itself a linear function of transmit powers of other users due to~\eqref{eqn_int}. In fact,~\eqref{opt_power} can be written as a linear equation system $\bo{Ap}=\bo{b}$ with unknown variables $\bo{p}$. In this case, $\bo{p}$ is a vector consisting of all transmit powers $p_{s,i}$ and $\bo{A}$ and $\bo{b}$ are the coefficients of the linear equation system given by~\eqref{opt_power}. To compute the optimal transmit powers, we solve the following simple optimization:
\begin{align}\label{eqn_siw}
	&\underset{\bo{p}}{\text{minimize}}
	\quad\|\bo{Ap}-\bo{b}\|_1\\
	& \text{subject to}\quad 0\leq p_{s,i}\leq P_{\max},\quad\forall s\in\cc{S},i\in\cc{U}_s.\nonumber
\end{align}

After finding the optimal powers from~\eqref{eqn_siw}, we can compute the uplink communication rates by using~\eqref{eqn_uls}. We then unschedule those users whose rates do not satisfy~\eqref{opt_5} and set their transmit power to zero.

\begin{algorithm}[t]
	\caption{Optimal scheduler with random DP noise (OptSched)}\label{alg_df2}
	\begin{algorithmic}[1]
		\State Initialize the values of $(\bo{R}_s,\bo{p}_s,\bs{\sigma}_{s})_{s\in\cc{S}}$ randomly such that they satisfy~\eqref{opt_1}-\eqref{opt_4}.
		\For{$s\in\cc{S}$}
		\State For fixed $(\bo{R}_{\tilde{s}},\bo{p}_{\tilde{s}})_{{\tilde{s}}\in\cc{S}\backslash\{s\}}$ and $(\bs{\sigma}_{s})_{s\in\cc{S}}$, obtain a \mbox{(sub-)optimal} resource block allocation matrix $\bo{R}^*_s$ by solving the  optimization problem in~\eqref{opt_ob3}.
		\EndFor
		\State Compute $(\bo{p}^*_s)_{s\in\cc{S}}$ by solving~\eqref{eqn_siw} and unschedule those users whose communication rates do not meet~\eqref{opt_5}.
		\State Output the resulting parameters as a (sub-)optimal solution $(\bo{R}^*_s,\bo{p}^*_s,\bs{\sigma}_{s})_{s\in\cc{S}}$.
	\end{algorithmic}
\end{algorithm}

Based on these solutions, we propose two procedures for user scheduling and DP noise adjustment. Algorithm~\ref{alg_non} presents a random scheduler (RndSched). Algorithm~\ref{alg_df2} provides an optimal scheduler (OptSched) based on~\eqref{opt_ob3}. Both algorithms benefit from the power allocation procedure based on~\eqref{eqn_siw} and both apply random DP noise to achieve privacy.

We note that one advantage of the optimal scheduler is that it is linear and therefore efficient from a practical point of view compared with~\eqref{eqn_jww}. Nevertheless, the drawback of this approach is that it is performed sequentially and cell by cell. As a result, there is no guarantee that this approach always provides us with an optimal solution. However, as we will see in Subsection~\ref{sec_num}, it delivers very good results compared with the randomized scheduler. In the next subsection, we extend this algorithm to include a DP optimizer.

\subsection{DP Optimizer}

Let the transmit powers and resource block allocation values from the optimal scheduler  $(\bo{R}^*_s,\bo{p}^*_s)_{s\in\cc{S}}$ be given. The DP noise optimization problem is then given by
\begin{align}
& \underset{(\bs{\sigma}_{s})_{s\in\cc{S}}}{\text{minimize}}
\quad \sum_{s\in\cc{S}}\sum_{i\in\cc{U}_s}\frac{\sum_{n=1}^{R} r_{s,i}^{(n)}}{K_{s,i}^2\sigma_{s,i}^2} \label{eqn_sil}\\
& \text{subject to }\eqref{opt_1}, \eqref{opt_6}.\nonumber
\end{align}

Since the objective function and all constraints in~\eqref{eqn_sil} are convex, the global optimal solution can be obtained by solving the Karush-Kuhn-Tucker (KKT)~\cite{bazaraa2013nonlinear} conditions.
The Lagrange function can be formulated as:
\begin{align*}
\mathcal L\big((\bs \sigma_{s})_{s\in\cc{S}}\big)
&=\sum_{s\in\cc{S}}\sum_{i\in\cc{U}_s}\frac{\sum_{n=1}^{R} r_{s,i}^{(n)}}{K_{s,i}^2\sigma_{s,i}^2}\\ &\quad+\kappa\bigg(\sum_{s\in\cc{S}}
\sum_{i\in\cc{U}_s}K_{s,i}\sigma_{s,i}^2\sum_{n=1}^{R}r_{s,i}^{(n)}\\&\qquad\qquad- V_{\max}\sum_{s\in\cc{S}}
\sum_{i\in\cc{U}_s}K_{s,i}\sum_{n=1}^{R}r_{s,i}^{(n)}\bigg),
\end{align*}
where $\kappa\geq0$ is a Lagrange multiplier.

Setting the derivative of $\mathcal L$ with respect to $\sigma_{s,i}$ to zero yields
\begin{align}
-\frac{2\sum_{n=1}^{R}r_{s,i}^{(n)}}{K_{s,i}^2\sigma_{s,i}^3}+2\kappa K_{s,i}\sum_{n=1}^{R}r_{s,i}^{(n)}\sigma_{s,i}=0.\label{eqn_iow}
\end{align}

Let $\sum_{n=1}^{R}r_{s,i}^{(n)}=1$ hold. It then follows by combing~\eqref{eqn_iow} with~\eqref{opt_6} that
\begin{align}\label{DP_opv1}
\sigma_{s,i}=\left.\left(K_{s,i}^3 \kappa\right)^{-\frac{1}{4}}\right|_{  \frac{N_{\min}}{K_{s,i}} },
\end{align}
where $a|_b=\max\{a,b\}, s=\cc{S},$ and $i\in\cc{U}_s$. 
If $\sum_{n=1}^{R}r_{s,i}^{(n)}=0$ holds, then we have $\sigma_{s,i}=0$.

\begin{figure}[!t]
	\centering
	\Large
	\includegraphics[scale=0.66]{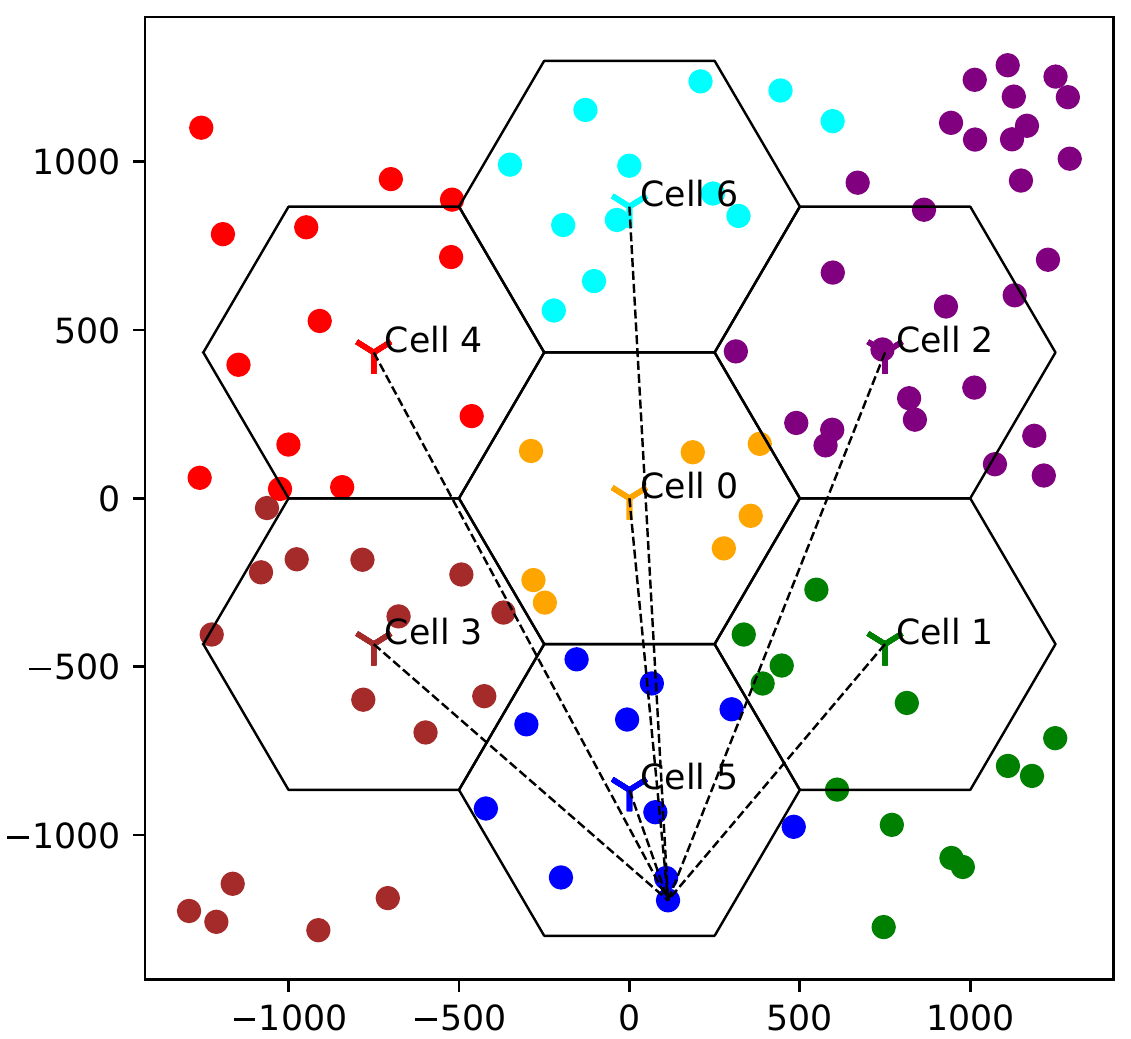}
	\caption{User and channel initialization with 100 users.}
	\label{fig_sim_ini}
\end{figure}

For the optimal solution, constraint \eqref{opt_1} always holds with equality since the objective function monotonically decreases with increasing $\sigma_{s,i}$. Consequently, substituting \eqref{DP_opv1} into \eqref{opt_1} with equality implies that
 \begin{align}\label{DP_opv2}
\sum_{s\in\cc{S}} \sum_{i\in\cc{U}_s}K_{s,i}&\sum_{n=1}^{R}r_{s,i}^{(n)}
\left.\left(K_{s,i}^3 \kappa \right)^{-\frac{1}{2}}\right|_{  \frac{N_{\min}^2}{K_{s,i}^2} }
\nonumber\\
&=V_{\max}\sum_{s\in\cc{S}}
\sum_{i\in\cc{U}_s}K_{s,i}\sum_{n=1}^{R}r_{s,i}^{(n)}.
\end{align}

After the value of $\kappa$ is found from~\eqref{DP_opv2}, the optimal $\sigma^*_{s,i}$ can be computed from~\eqref{DP_opv1}.
We then combine this scheme with the procedure in Subsection~\ref{sub_rba}. A summary of this scheme is provided in Algorithm~\ref{alg_df3} (OptSched+DP).

\section{Simulations and Numerical Solutions}

\subsection{Optimization Problems}\label{sec_num}
In this subsection, we present the numerical solutions of the algorithms that were presented in Section~\ref{sec_alg}. In this regard, we apply the Python optimization packages CVXPY, CVXOPT, GLPK, and ECOS~\cite{diamond2016cvxpy,agrawal2018rewriting,cvxopt,glpk,Domahidi2013ecos}.

Since Algorithms~\ref{alg_df2} and~\ref{alg_df3} are heuristic, their solutions depend on the initial values of the optimizing variables $(\bo{R}_s,\bo{p}_s,\bs{\sigma}_{s})_{s\in\cc{S}}$ as well as the wireless channels and the number of training samples at each user. As a result, we repeat the computations for several random initial values, channels and data distributions among the users and then compute the average.

\begin{table}
	\caption{}
	\centering
	\begin{tabular}{|l|r|}
		\hline\hline
		System parameter& values\\ 
		\hline
		Number of cells or base stations $(S)$:& $7$\\
		Total number of users:& 100\\
		Cell radius:& 500m\\
		Uplink center frequency:& 2450MHz\\
		Channels' Rayleigh distribution scale parameter:& 1\\
		Uplink resource block bandwidth (B):& $180 \rr{KHz}$\\
		Thermal noise power spectral density $(N_0)$:& $-174 \text{dBm}$\\
		Maximum transmit power $(P_{\rr{max}})$:& $10 \text{dBm}$\\
		Minimum communication rate $(R_{\rr{min}})$:& $100 \rr{Kbs}$\\
		DP noise error upper bound $(V_{\rr{max}})$:& $12$\\
		Minimum total DP noise at each user $(N_{min})$:& $100$\\
		\hline
	\end{tabular}\label{tab_par}
\end{table}

\begin{algorithm}[t]
	\caption{Optimal scheduler with DP noise optimizer (OptSched+DP)}\label{alg_df3}
	\begin{algorithmic}[1]
		\State Perform  Algorithm~\ref{alg_df2}. 
		\State For the given $(\bo{R}^*_s,\bo{p}^*_s)_{s\in\cc{S}}$ from  Algorithm~\ref{alg_df2}, obtain the optimal $(\bs{\sigma}^*_{s})_{s\in\cc{S}}$ by solving~\eqref{eqn_sil}.
		\State Output the resulting parameters as a (sub-)optimal solution $(\bo{R}^*_s,\bo{p}^*_s,\bs{\sigma}^*_{s})_{s\in\cc{S}}$.
	\end{algorithmic}
\end{algorithm}

\begin{figure*}
	\begin{subfigure}{0.49\textwidth}
		\includegraphics[width=1\linewidth]{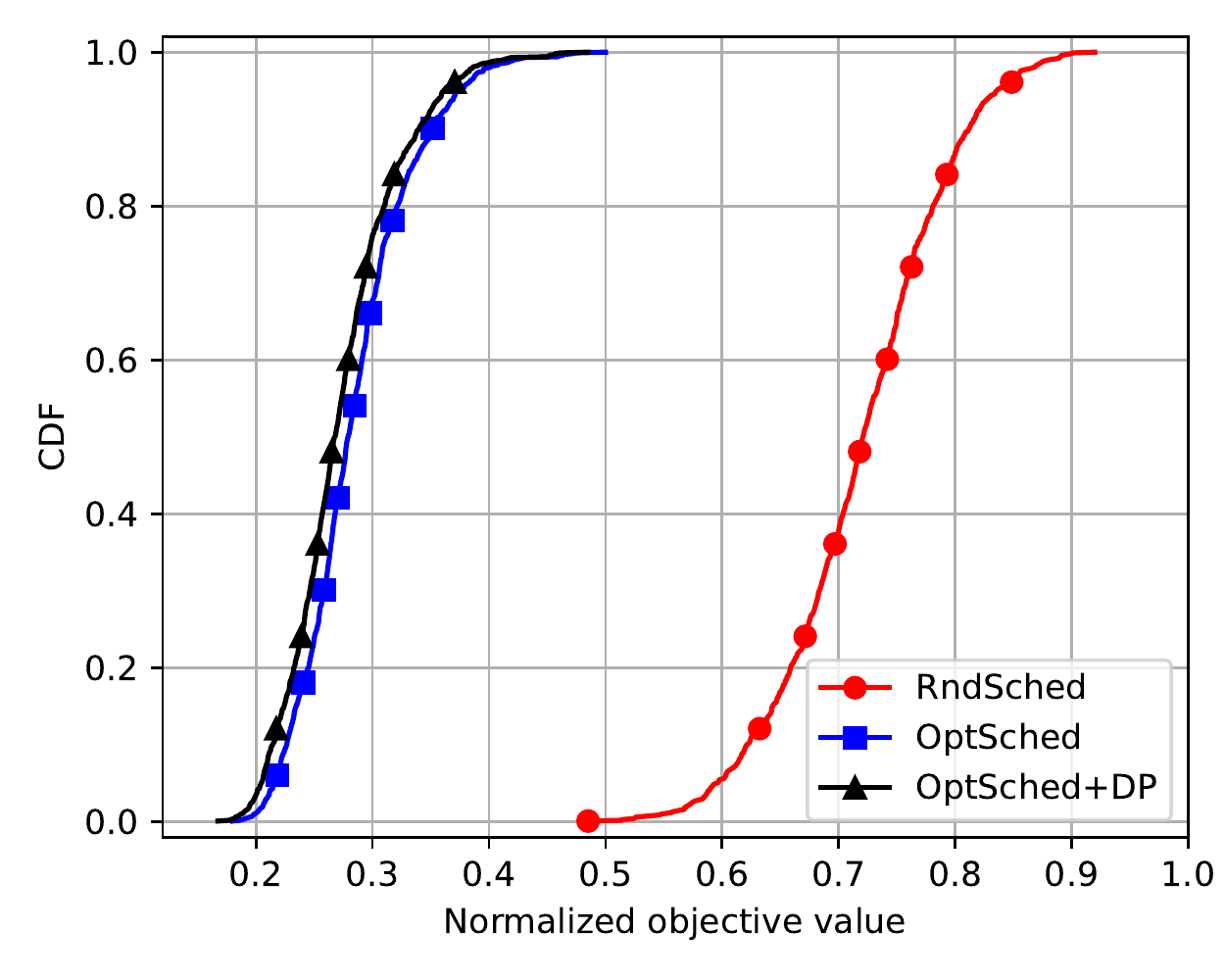}
		\caption{$R=5,\gamma=10^6$}\label{fig_sim_opt1}
	\end{subfigure}
	\hspace*{\fill}  
	\begin{subfigure}{0.49\textwidth}
		\includegraphics[width=1\linewidth]{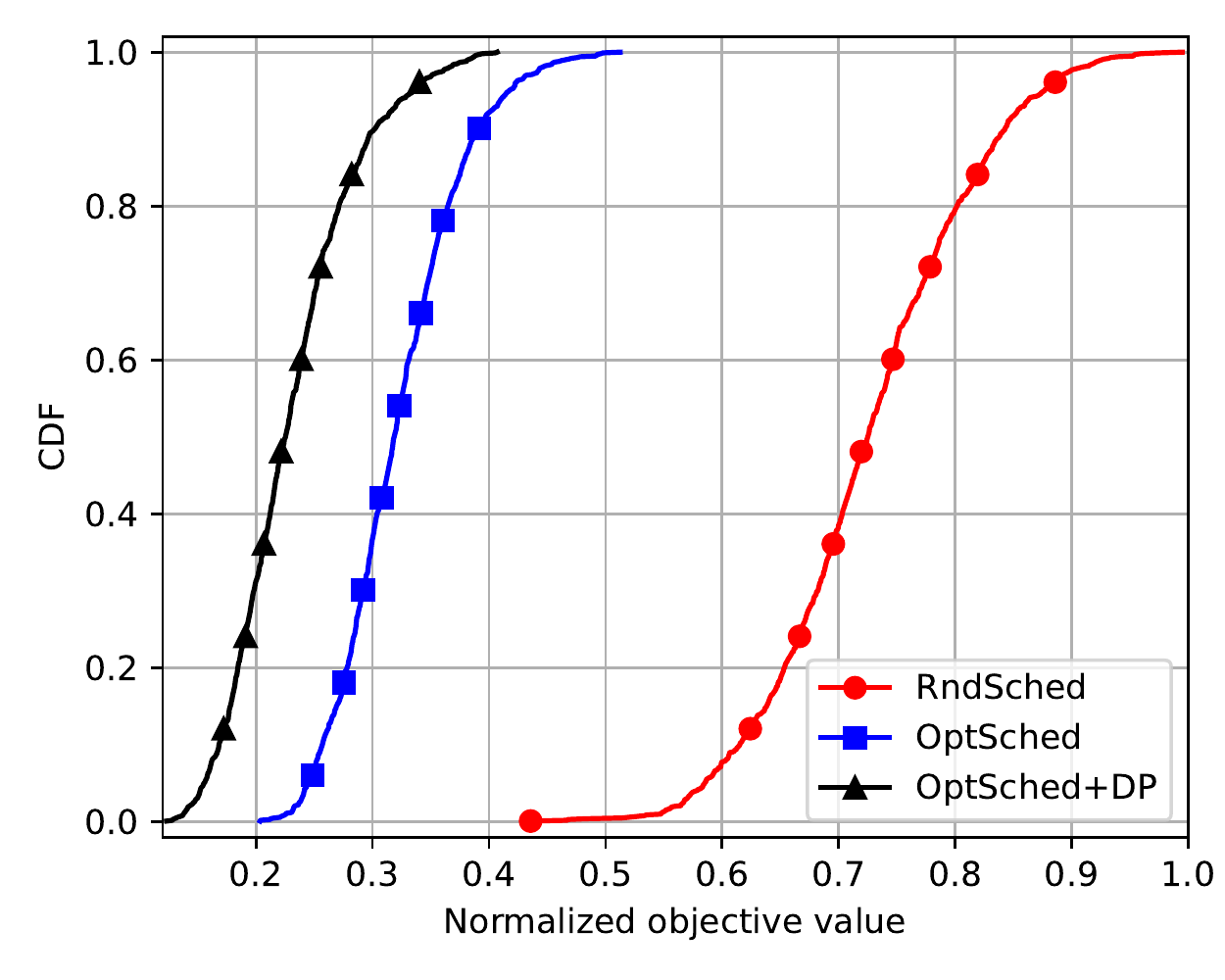}%
		\caption{$R=8,\gamma=10^7$}\label{fig_sim_opt2}
	\end{subfigure}
	
	\caption{Numerical results of the optimization problem in~\eqref{opt_ob1}.} \label{fig_sim_opt}
\end{figure*}

To this end, the variables $(\bo{R}_s)_{s\in\cc{S}}$ are first initialized based on a shuffled Round-robin scheme and $(\bo{p}_s,\bs{\sigma}_{s})_{s\in\cc{S}}$ are set uniformly at random such that $0\leq p_{s,i}\leq P_{\max}$ and $N_{\min}/K_{s,i}\leq \sigma_{s,i}\leq 6N_{\min}/K_{s,i}$ hold. 
Second, the users are positioned in a square area consisting of seven hexagon cells according to a uniform distribution. The edge devices are then assigned to their nearest base stations according to their random position. Based on their distances to the base stations, their fading channels are then computed by applying~\eqref{eqn_cgn}. 

An example of channel initialization, which is generated by our simulator in Python language, is shown in Fig.~\ref{fig_sim_ini}. In this case, the channels between one of the users and base stations are depicted as dashed lines. We notice that the cells 1-6 in this setting can cover users also outside their area while the central cell only covers devices inside the central hexagon. As a result, the effects of boundary and central cells are both taken into account in our simulations.

After the users are assigned to their corresponding base stations, the training data is randomly distributed among all users. Inspired by~\cite{MLSYS2020}, the number of samples $K_{s,i}$ are determined by a lognormal distribution.

Algorithms~\ref{alg_df2} and~\ref{alg_df3} should then provide us with $(\bo{R}^*_s,\bo{p}^*_s)_{s\in\cc{S}}$ and $(\bo{a}^*_s)_{s\in\cc{S}}$ which determine (sub-)optimal allocated resource blocks, uplink transmit powers, and the scheduled users.

The system parameters that are used in the computations are listed in Table~\ref{tab_par}. Fig.~\ref{fig_sim_opt} shows the results of all algorithms in the form of an empirical Cumulative Distribution Function (CDF) of the normalized objective value in~\eqref{opt_ob1}. The normalization is done by dividing the value of the objective function by the total number of samples (scheduled or unscheduled). The CDF is computed for two values of available number of resource blocks $R$ and the optimization constant $\gamma$ from~\eqref{opt_ob1} .
The results are averaged over $10^3$ random channels and initial values. As seen in Fig.~\ref{fig_sim_opt}, the OptSched (Algorithm~\ref{alg_df2}) outperforms the RndSched (Algorithm~\ref{alg_non}) in terms of minimizing the objective value in~\eqref{opt_ob1}. Moreover, the OptSched+DP (Algorithm~\ref{alg_df3}) further improves the results of the OptSched by reducing the total privacy leakage. Furthermore, the OptSched+DP achieves lower values for $\gamma=10^7$ compared with the case in which $\gamma=10^6$. This is because, larger $\gamma$ in~\eqref{opt_ob1} gives more weight to the DP noise optimization.

We also notice that by increasing $R$ from 5 to 8, the normalized objective values of the RndSched get slightly closer to the outcome of the OptSched algorithm. This is due to the fact that by increasing $R$ and keeping the total number of users constant, chances that all users are successfully scheduled by RndSched become higher. In this case, for large values of $R$, the RndSched might eventually achieves the same performance as for OptSched. However, the choice of selecting a large number of resource blocks for a low number of users is not desirable due to the limited amount of available bandwidth.

\subsection{Federated Learning Simulations}\label{sec_sim}

In this subsection, we apply the random parameters $\bo{a}_s$ and $\bs{\sigma}_{s}$ as well as (sub-)optimal $\bo{a}^*_s$ and $\bs{\sigma}^*_{s}$ from Subsection~\ref{sec_num} to an FL system as described in Algorithm~\ref{alg_dfl}. In this case, we assume that the main server and all users each maintain a fully connected neural network in the form of a multi-label classifier. The networks consist of two hidden layers, each with 256 nodes. To implement the simulations, we apply the TensorFlow, NumPy, and Matplotlib   packages~\cite{tensorflow2015-whitepaper,harris2020array,Hunter:2007}. Furthermore, we use the MNIST image dataset~\cite{lecun1998mnist} to train and test the multi-label classifier.  

\begin{figure*}
	\begin{subfigure}{0.49\textwidth}
		\includegraphics[width=1\linewidth]{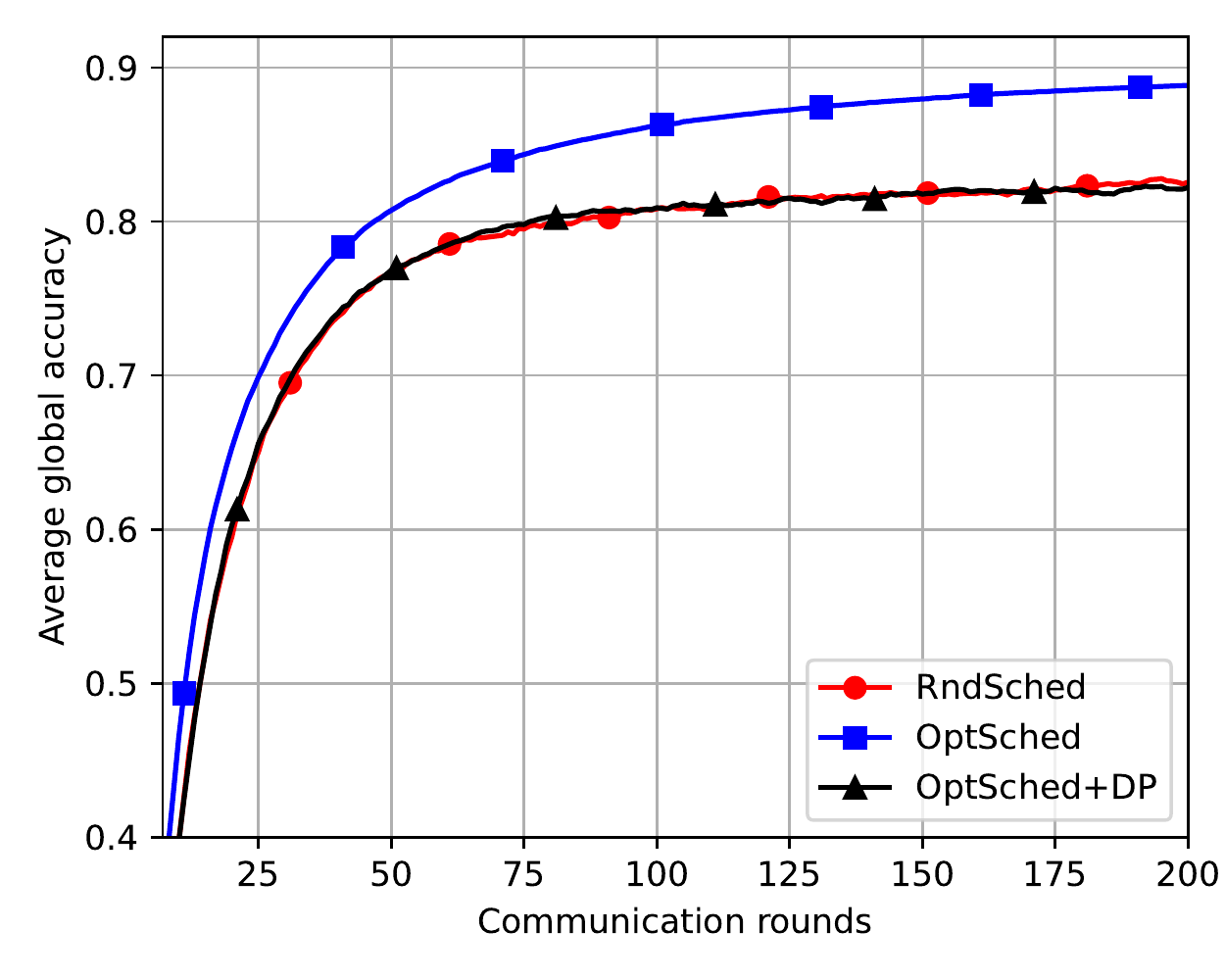}
		\caption{$R=5, \gamma=10^6$}\label{fig_sim_tf:a}
	\end{subfigure}
	\hspace*{\fill}
	\begin{subfigure}{0.49\textwidth}
		\includegraphics[width=1\linewidth]{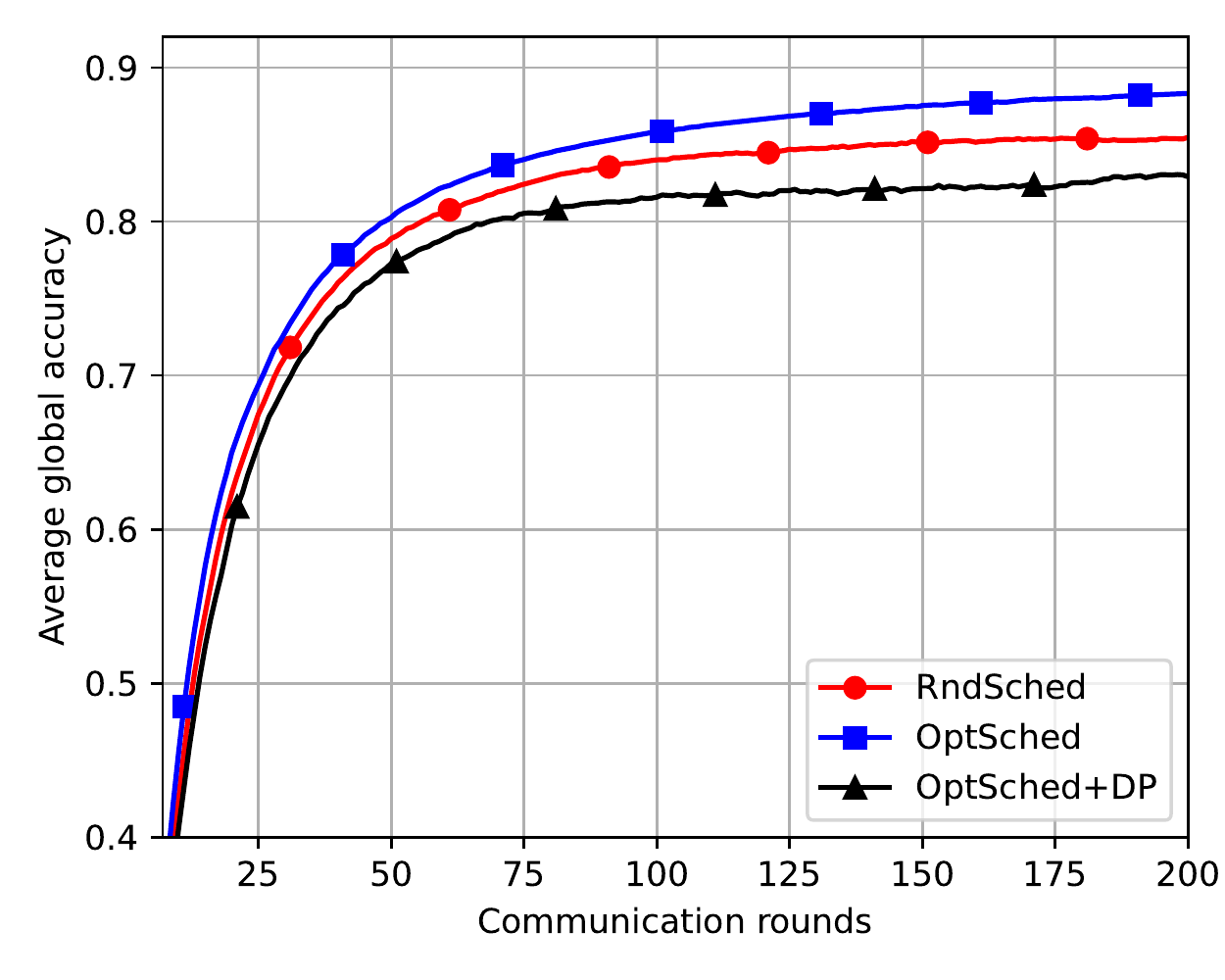}%
		\caption{$R=8, \gamma=10^7$}\label{fig_sim_tf:b}
	\end{subfigure}
	
	\bigskip
	
	\begin{subfigure}{0.49\textwidth}
		\includegraphics[width=1\linewidth] {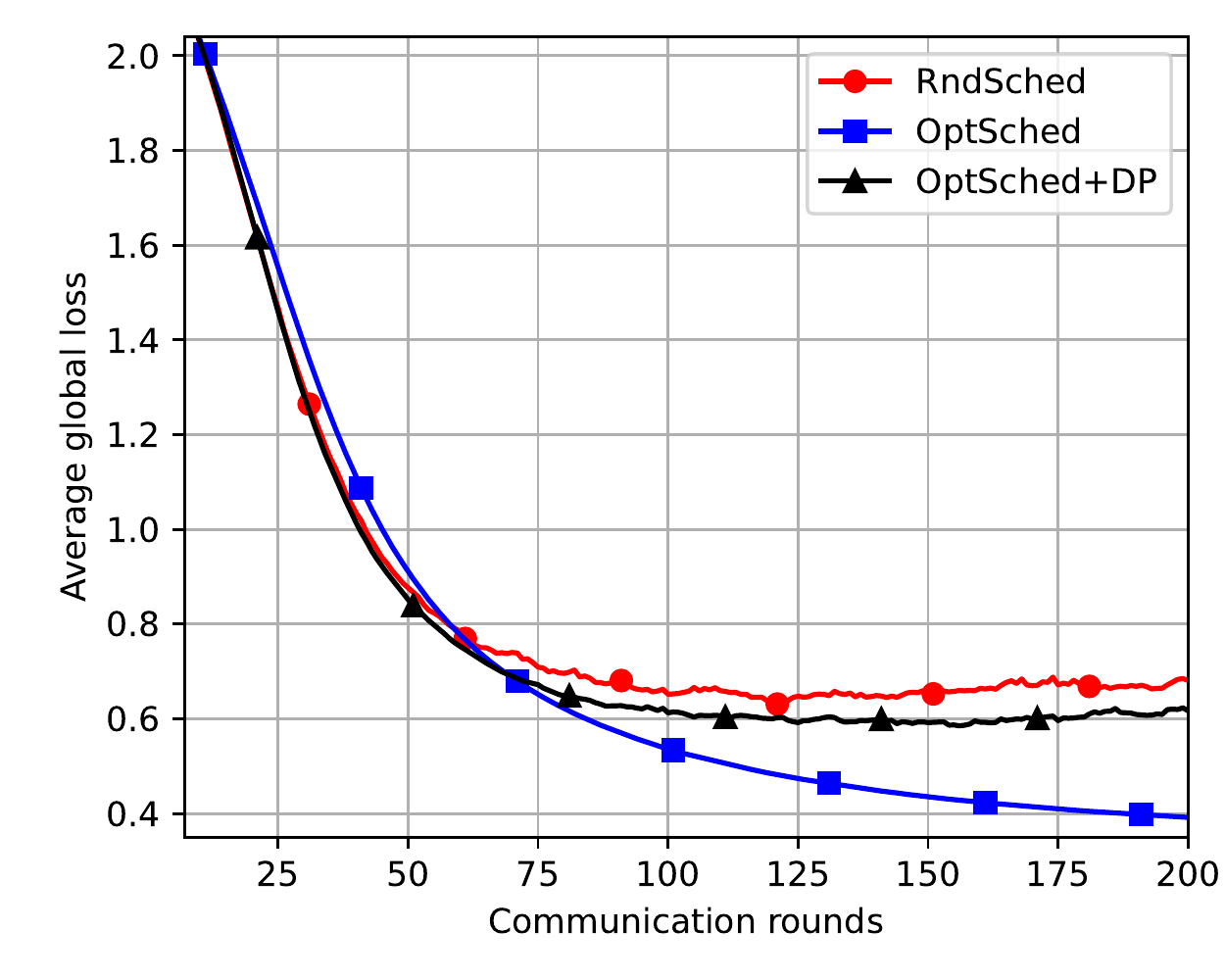}
		\caption{$R=5, \gamma=10^6$}\label{fig_sim_tf:c}
	\end{subfigure}
	\hspace*{\fill} 
	\begin{subfigure}{0.49\textwidth}
		\includegraphics[width=1\linewidth] {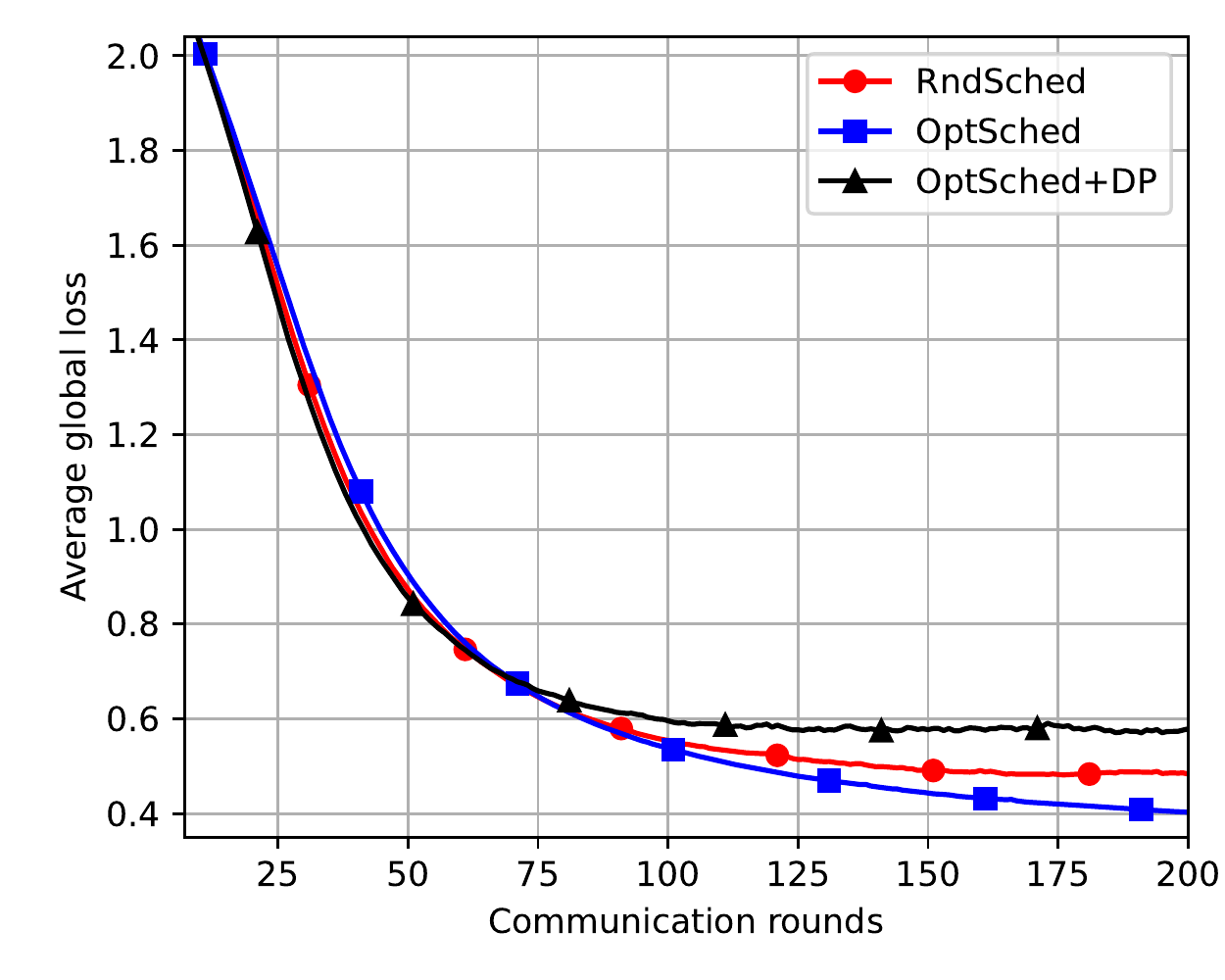}%
		\caption{$R=8, \gamma=10^7$}\label{fig_sim_tf:d}
	\end{subfigure}
	
	\bigskip
	
	\begin{subfigure}{0.49\textwidth}
		\includegraphics[width=1\linewidth] {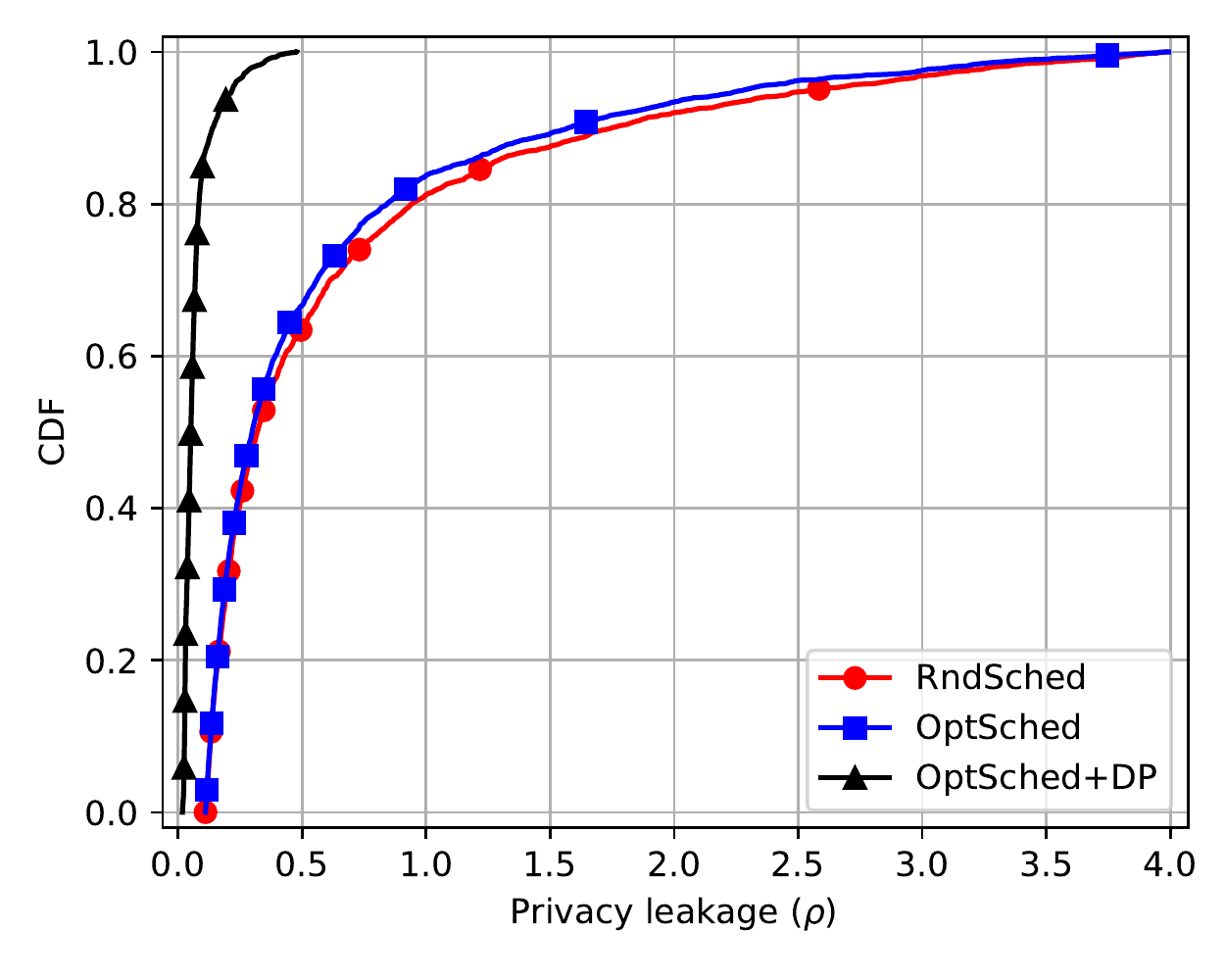}
		\caption{$R=5, \gamma=10^6$}\label{fig_sim_tf:e}
	\end{subfigure}
	\hspace*{\fill} 
	\begin{subfigure}{0.49\textwidth}
		\includegraphics[width=1\linewidth] {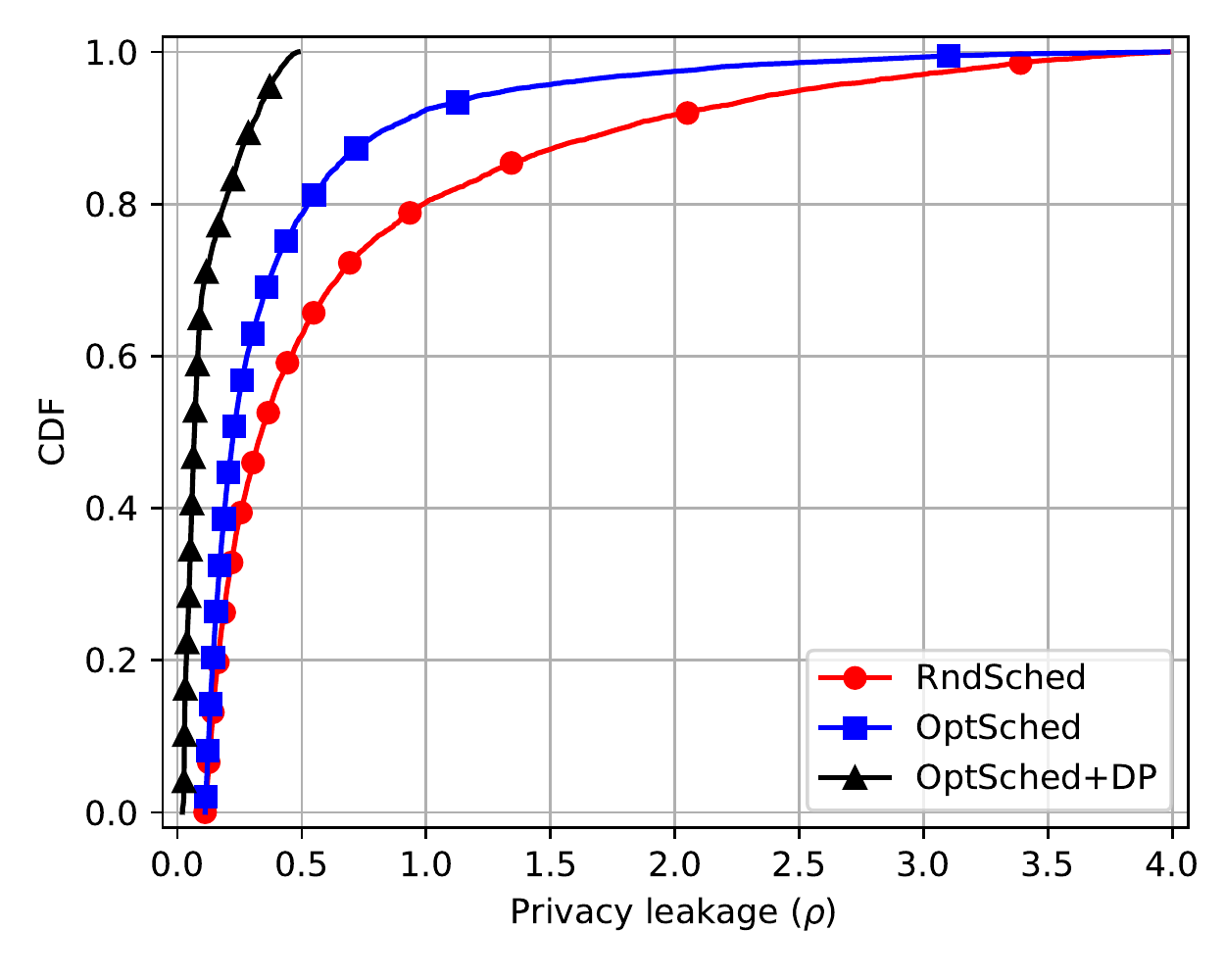}%
		\caption{$R=8, \gamma=10^7$}\label{fig_sim_tf:f}
	\end{subfigure}
	
	\caption{FL system with a learning rate of $\lambda=0.05$ and a maximum gradient global norm of $L=10$.} \label{fig_sim_tf}
\end{figure*}

In this setting, we train the local models over $T=200$ communication rounds between users and the main server. To follow our mathematical model in Section~\ref{sec_mod}, we perform no local iterations and use the batch gradient descent scheme. We do not apply any decay and use a fixed learning rate $\lambda=0.05$. 

Furthermore, to guarantee that Assumption~\ref{asm_inp} holds, the gradient of all weights are clipped so that their global norm is smaller than or equal to $L=10$. This directly affects the amount of privacy leakage as given by~\eqref{eqn_sli}. 

We perform the simulations over 100 channels and initial values and then average the resulting accuracy and loss. Furthermore, we generate the empirical CDF of the privacy leakage of all users. Fig.~\ref{fig_sim_tf} shows the accuracy, loss, and privacy leakage CDF of this learning system for different values of available resource blocks $R$ and optimization constant $\gamma$.

As seen in Fig.~\ref{fig_sim_tf}, the OptSched outperforms the RndSched algorithm in terms of accuracy and loss for both $R=5$ and $R=8$. In this case, the OptSched systematically selects the users with large chunks of data that have a better channel and suffer less from the inter-cell interference. The RndSched algorithm, however, fails in this scenario since it applies a random scheduling scheme. 

The OptSched+DP, on the other hand, slightly degrades the performance of the optimal scheduler by increasing and optimizing the DP noise. Yet OptSched+DP provides a similar or even better performance compared with RndSched scheme for small values of $R$ (see Figs~\ref{fig_sim_tf:a} and~\ref{fig_sim_tf:c}). The degradation is the price that is paid to improve the privacy. Figs.~\ref{fig_sim_tf:e} and~\ref{fig_sim_tf:f} show the empirical CDF of the privacy leakage ($\rho$). In this case, the value of $\rho$ at each user is computed by using~\eqref{eqn_sli} and the results over all simulation iterations are collected to compute the CDF. The simulations show that the OptSched+DP scheme substantially reduces the amount of privacy leakage at each user. In particular it achieves a maximum privacy leakage of around $\rho=0.5$ thanks to the DP optimizer scheme. This is a significant improvement compared with the RndSched scheme with a maximum leakage of around $\rho=4$.

Adjusting the optimization constant $\gamma$ is also crucial. In this regard, by choosing $\gamma=10^7$ the OptSched achieves lower privacy leakage compared with $\gamma=10^6$ (see Figs.~\ref{fig_sim_tf:e} and~\ref{fig_sim_tf:f}). This is because larger $\gamma$ gives less weight to the optimal scheduler in~\eqref{opt_ob1} and users with higher DP noise power are preferred in scheduling.

\section{Conclusion}

In this work, a privacy preserving FL procedure in a multiple base station scenario with inter-cell interference has been considered. An upper bound on the optimality gap of the convergence term of this learning scheme has been derived and an optimization problem to reduce this upper bound has been provided. We have proposed two sequential algorithms to obtain (sub-)optimal solutions for this optimization task; namely an optimal scheduler (OptSched) in Algorithm~\ref{alg_df2} and its extended version with DP optimizer (OptSched+DP) in Algorithm~\ref{alg_df3}. In designing these schemes we avoid non-linearity in the integer programming problems. The outputs of these algorithms are then applied to an FL system.

Simulation results have shown that the OptSched increases the accuracy of the classification FL system and reduces the loss compared with the RndSched when the number of available resource blocks $R$ is small. In this case, when the total number of users is $K=100$ and $R=5$, the OptSched shows an accuracy improvement of over $6\%$. Simulations have further shown that the OptSched not only improves the accuracy but also can reduce the privacy leakage compared with the RndSched if the parameter $\gamma$ is set properly.

The OptSched+DP, on the other hand, further optimizes the DP noise and substantially reduces the privacy leakage compared with both RndSched and OptSched. In this case, simulations have shown that the OptSched+DP reduces the maximum privacy leakage for both $R=5$ and $R=8$ by a factor of 8 (from $\rho=4$ to $\rho=0.5$). It is worth mentioning that when $R$ is small (e.g. $R=5$), this improvement is achieved while OptSched+DP shows a similar or even better performance in terms of accuracy and loss compared with RndSched. 

\section*{Appendix}
\subsection*{Proof of Theorem~\ref{thm_con}\centering}
\begin{proof}
It follows by using Assumption~\ref{asm_twc} and applying the Taylor expansion to the global loss function $f$ that
\begin{align}
f(\bo{w}^{(t+1)})\leq f(\bo{w}^{(t)})+\big(&\bo{w}^{(t+1)}-\bo{w}^{(t)}\big)^\intercal\nabla f(\bo{w}^{(t)})\nonumber\\&+\frac{L}{2}\big\|\bo{w}^{(t+1)}-\bo{w}^{(t)}\big\|_2^2.\label{eqn_oas}
\end{align}

Next, we compute the local updates at each user $i\in\cc{U}_s$ by combing~\eqref{eqn_sdf},~\eqref{eqn_upd}, and~\eqref{eqn_lcm} as follows
\begin{align}
\bo{w}_{s,i}^{(t+1)}=\bo{w}^{(t)}- \frac{\lambda}{K_{s,i}}\sum_{k=1}^{K_{s,i}}\nabla l(\bo{w}^{(t)}, \bo{x}_{s,i}^{(k)})-\lambda\bo{n}^{(t)}_{s,i}\label{eqn_sop}.
\end{align}
Furthermore, Combining~\eqref{eqn_agg} and~\eqref{eqn_ag2} implies that
\begin{align}
\bo{w}^{(t+1)}=\frac{1}{K_{\rr{a}}}\sum_{s\in\cc{S}}\sum_{i\in\cc{U}_s}K_{s,i}a_{s,i}\bo{w}_{s,i}^{(t+1)}.\label{eqn_so2}
\end{align}

We then obtain the global update at the main server by inserting the value of $\bo{w}_{s,i}^{(t+1)}$ from~\eqref{eqn_sop} into~\eqref{eqn_so2} as follows
\begin{align}
\bo{w}^{(t+1)}-\bo{w}^{(t)}=-\frac{\lambda}{K_{\rr{a}}}&\sum_{s\in\cc{S}}\sum_{i\in\cc{U}_s}a_{s,i}\sum_{k=1}^{K_{s,i}}\nabla l(\bo{w}^{(t)}, \bo{x}_{s,i}^{(k)})\nonumber\\
&-\frac{\lambda}{K_{\rr{a}}}\sum_{s\in\cc{S}}\sum_{i\in\cc{U}_s}K_{s,i}a_{s,i}\bo{n}^{(t)}_{s,i}.\label{eqn_sla}
\end{align}

To simplify the rest of calculations, we define a new random variable to reflect the difference between the global update and the global gradient as below:
\begin{align}
\Delta^{(t)}\coloneqq \nabla f(\bo{w}^{(t)}) +\frac{1}{\lambda}\big(\bo{w}^{(t+1)}-\bo{w}^{(t)}\big).\label{eqn_whm}
\end{align}

Now by inserting the term $\bo{w}^{(t+1)}-\bo{w}^{(t)}$ (the global update) from~\eqref{eqn_whm} into~\eqref{eqn_oas}, we have that
\begin{align}
f(\bo{w}^{(t+1)})\leq f(\bo{w}^{(t)})&+\lambda\big(\Delta^{(t)}-\nabla f(\bo{w}^{(t)})\big)^\intercal\nabla f(\bo{w}^{(t)})\nonumber\\&+\frac{\lambda^2 L}{2}\big\|\Delta^{(t)}-\nabla f(\bo{w}^{(t)})\big\|_2^2.\label{eqn_wld}
\end{align}

Furthermore, the following identity always holds:
\begin{align}
\bo{\|\bo{u}-\bo{v}\|_2^2}=\|\bo{u}\|_2^2+\|\bo{v}\|_2^2-2
\bo{u}^\intercal\bo{v}.\label{eqn_idl}
\end{align}

Considering the learning step size to be $\lambda=1/L$ and applying the identity~\eqref{eqn_idl} to~\eqref{eqn_wld}, it follows that
\begin{align}
f(\bo{w}^{(t+1)}) - f(\bo{w}^*)&\leq f(\bo{w}^{(t)})-f(\bo{w}^*)\nonumber\\&+\frac{1}{2L}\Big[\|\Delta^{(t)}\|_2^2-\|\nabla f(\bo{w}^{(t)})\|_2^2\Big],\label{eqn_ska}
\end{align}
where $f(\bo{w}^*)$ is the optimal loss function (Assumption~\ref{eqn_khs}).

Inspired by~\cite{chen2020joint}, We first obtain an upper bound on the expectation of the term $\|\Delta^{(t)}\|_2^2$ on the right hand side of~\eqref{eqn_ska}. It follows by combining~\eqref{eqn_sla} and~\eqref{eqn_whm} that
\begin{align}
&\bb{E}\Big[\|\Delta^{(t)}\|_2^2\Big]\nonumber\\&=\bb{E}\Bigg[ \bigg\|\nabla f(\bo{w}^{(t)})-\frac{1}{K_{\rr{a}}} \sum_{s\in\cc{S}}\sum_{i\in\cc{U}_s}a_{s,i}\sum_{k=1}^{K_{s,i}}\nabla l(\bo{w}^{(t)}, \bo{x}_{s,i}^{(k)})\nonumber\\
&\qquad\qquad\qquad-\frac{1}{K_{\rr{a}}}\sum_{s\in\cc{S}}
\sum_{i\in\cc{U}_s}K_{s,i}a_{s,i}\bo{n}^{(t)}_{s,i}\bigg\|_2^2 \Bigg]\nonumber\\
&= \bb{E}\Bigg[\bigg\|
-\frac{K-K_{\rr{a}}}{K K_{\rr{a}}}\sum_{\substack{(s,i):\\a_{s,i}=1}}\sum_{k=1}^{K_{s,i}}
\nabla l(\bo{w}^{(t)}, \bo{x}_{s,i}^{(k)})\nonumber\\
&\qquad\qquad\qquad+\frac{1}{K}\sum_{\substack{(s,i):\\a_{s,i}=0}}\sum_{k=1}^{K_{s,i}}
\nabla l(\bo{w}^{(t)}, \bo{x}_{s,i}^{(k)})\bigg\|_2^2
\Bigg]\nonumber\\
&\qquad\qquad\qquad+\bb{E}\Bigg[\bigg\|\frac{1}{K_{\rr{a}}}\sum_{s\in\cc{S}}
\sum_{i\in\cc{U}_s}K_{s,i}a_{s,i}\bo{n}^{(t)}_{s,i}\bigg\|_2^2
\Bigg]\label{eqn_slx}\\
&\leq \bb{E}\Bigg[\bigg(
\frac{K-K_{\rr{a}}}{K K_{\rr{a}}}\sum_{\substack{(s,i):\\a_{s,i}=1}}\sum_{k=1}^{K_{s,i}}\|
\nabla l(\bo{w}^{(t)}, \bo{x}_{s,i}^{(k)})\|_2\nonumber\\
&\qquad\qquad\qquad+\frac{1}{K}\sum_{\substack{(s,i):\\a_{s,i}=0}}\sum_{k=1}^{K_{s,i}}\|
\nabla l(\bo{w}^{(t)}, \bo{x}_{s,i}^{(k)})\|_2\bigg)^2
\Bigg]\nonumber\\
&\qquad\qquad\qquad+\sum_{e=1}^{d}\bb{E}\Bigg[\bigg(\sum_{s\in\cc{S}}
\sum_{i\in\cc{U}_s}\frac{K_{s,i}a_{s,i}}{K_{\rr{a}}}n_{s,i,e}^{(t)}\bigg)^2\Bigg],\label{eqn_ji2}
\end{align}
where~\eqref{eqn_slx} follows by applying~\eqref{eqn_bms} and~\eqref{eqn_idl} and the fact that the DP noise is independent of other random variables and $\bb{E}[\bo{n}^{(t)}_{s,i}]=\bo{0}$. Inequality~\eqref{eqn_ji2} is due to the triangle inequality and the fact that the vectors $\bo{n}_{s,i}=(n_{s,i,e}^{(t)})_{e\in[d]}$ are $d$-dimensional.

Next, by applying Assumption~\ref{asm_tw1} to~\eqref{eqn_ji2} we have for some $\xi_1\geq 0$ and $\xi_2\geq 1$ that
\begin{align}
&\bb{E}\Big[\|\Delta^{(t)}\|_2^2\Big]\nonumber\\&\leq
	\bigg(\sum_{s\in\cc{S}}\sum_{i\in\cc{U}_s}\frac{2K_{s,i}}{K}
	\Big(1-a_{s,i}\Big)\bigg)^2\bb{E}\bigg[\Big(\xi_1+\xi_2\|\nabla f(\bo{w}^{(t)})\|_2^2\Big)\bigg]\nonumber\\
	&\qquad\qquad\qquad+d\sum_{s\in\cc{S}}
	\sum_{i\in\cc{U}_s}\bigg(\frac{K_{s,i}a_{s,i}}{K_{\rr{a}}}\sigma_{s,i}\bigg)^2,\label{eqn_jik}
\end{align}
where the last term in~\eqref{eqn_jik} is obtained due to the fact that the random variables $n_{s,i,e}^{(j)}$ in~\eqref{eqn_ji2} are independent of each other.

On the other hand, since $f$ is $\mu$-strongly convex (Assumption~\ref{asm_scx}) we have that
\begin{align}
\|\nabla f(\bo{w}^{(t)})\|_2^2 &\geq 2\mu \big[f(\bo{w}^{(t)})-f(\bo{w}^*)\big].\label{eqn_mls}
\end{align}
By inserting~\eqref{eqn_jik} in~\eqref{eqn_ska} and using~\eqref{eqn_mls}, the proof follows.
\end{proof}

\bibliographystyle{IEEEtran}
\bibliography{ml_ref.bib}

\end{document}